\documentclass[letterpaper, 10 pt, conference]{IEEEconf}  % Comment this line out
\IEEEoverridecommandlockouts                              % This command is only
\usepackage{graphics,graphicx,psfrag,array,eepic}
\usepackage{multirow}% http://ctan.org/pkg/multirow
\usepackage{hhline}% http://ctan.org/pkg/
\usepackage{mathtools}
\usepackage[strict]{changepage}
\usepackage{epsfig}
\usepackage{amsmath}
\usepackage{amssymb}

\usepackage{float}
\restylefloat{table}

\usepackage{multicol}
\usepackage{color}
\usepackage{balance}
\usepackage{cite}
\usepackage{booktabs} % For formal tables
\usepackage{caption,booktabs}
\usepackage[tableposition=top]{caption}
\usepackage{subcaption}

\usepackage{lipsum}
\usepackage{comment}
\usepackage{amssymb}
\usepackage{graphicx}

\usepackage{comment}
\usepackage{bm}

\usepackage{algorithm}
\usepackage[noend]{algpseudocode}

\usepackage{mathtools, cuted}
\usepackage[thinc]{esdiff}

\usepackage{psfrag}
\usepackage{array}
\usepackage{latexsym,theorem}
\usepackage{amsfonts,amssymb,amsbsy}
\usepackage{tikz}
\usetikzlibrary{decorations.shapes,shapes.geometric,shadows,arrows,automata,positioning,calendar,mindmap,backgrounds,scopes,chains,er,3d,matrix,fit}
\usepackage{tikz-cd}
\usepackage{pgfplots}
\usetikzlibrary{intersections, pgfplots.fillbetween}

\DeclareMathOperator*{\argmin}{argmin}

\newcommand{\R}{{\rm I\!R}}

\newcommand{\N}{{\rm I\!N}}

\newcommand{\cO}{\mathcal{O}}

\newcommand{\cU}{\mathcal{U}}
\newcommand{\cX}{\mathcal{X}}

\newcommand{\cE}{\mathcal{E}}

\newcommand{\cT}{\mathcal{T}}
\newcommand{\cK}{\mathcal{K}}

\newcommand{\cV}{\mathcal{V}}

\newcommand{\cD}{\mathcal{D}}
\newcommand{\cW}{\mathcal{W}}
\newcommand{\cY}{\mathcal{Y}}
\newcommand{\cG}{\mathcal{G}}

\newcommand{\diag}{\operatorname{diag}}

\newtheorem{lemma}{Lemma}
\newtheorem{definition}{Definition}
\newtheorem{theorem}{Theorem}
\newtheorem{proposition}{Proposition}

\newtheorem{assumption}{Assumption}
\newtheorem{remark}{Remark}

\usepackage[british]{babel}
\usepackage{cite}
\usepackage{amsmath,amssymb,amsfonts}
\usepackage{graphicx}
\usepackage{textcomp}
\usepackage{mathrsfs}
\usepackage{url}
\usepackage{bm}
\usepackage{nicefrac}
\allowdisplaybreaks
\usepackage{alphalph,etoolbox}
\usepackage{mathtools,subcaption,float}
\usepackage{pgfplots}
\usepgfplotslibrary{colorbrewer}
\pgfplotsset{compat = 1.15} 
\usetikzlibrary{pgfplots.statistics, pgfplots.colorbrewer} 
\usepackage{pgfplotstable}

\title{\LARGE \bf
Probabilistic tube-based control synthesis of stochastic multi-agent systems under signal temporal logic
}

\author{Eleftherios E. Vlahakis$^1$, Lars Lindemann$^2$, Pantelis Sopasakis$^3$ and Dimos V. Dimarogonas$^1$
\thanks{This work was supported by the Swedish
Research Council (VR), the Knut \& Alice Wallenberg Foundation (KAW), the Horizon Europe Grant SymAware and the ERC Consolidator Grant LEAFHOUND. $^1$Division of Decision and Control Systems, School of Electrical Engineering and Computer Science, KTH Royal Institute of Technology, 10044, Stockholm, Sweden. Email: {\tt\small\{vlahakis,dimos\}@kth.se}. $^2$Thomas Lord Department of Computer Science, Viterbi School of Engineering, University of Southern California,  Los Angeles, 90089, CA, USA. Email: {\tt\small llindema@usc.edu}. $^3$School of Electronics, Electrical Engineering and Computer Science, Queen's University Belfast, Northern Ireland, UK. Email: {\tt\small p.sopasakis@qub.ac.uk}
}%
}

\begin{document}

\maketitle
\thispagestyle{empty}
\pagestyle{empty}

%%%%%%%%%%%%%%%%%%%%%%%%%%%%%%%%%%%%%%%%%%%%%%%%%%%%%%%%%%%%%%%%%%%%%%%%%%%%%%%%%%%
\begin{abstract}
We consider the control design of stochastic discrete-time linear multi-agent systems (MASs) 
under a global signal temporal logic (STL) specification to be satisfied at a predefined 
probability. 
By decomposing the dynamics into deterministic and error components, we construct a probabilistic 
reachable tube (PRT) as the Cartesian product of reachable sets of the individual error systems 
driven by disturbances lying in confidence regions (CRs) with a fixed probability. 
By bounding the PRT probability with the specification probability, we tighten all state 
constraints induced by the STL specification by solving tractable optimization problems 
over segments of the PRT, and convert the underlying stochastic problem into a deterministic one. 
This approach reduces conservatism compared to tightening guided by the STL structure. 
Additionally, we propose a recursively feasible algorithm to attack the resulting problem 
by decomposing it into agent-level subproblems, which are solved iteratively according to a 
scheduling policy. 
We demonstrate our method on a ten-agent system, where existing approaches are impractical.

\end{abstract}

% \begin{IEEEkeywords}
% multi-agent systems, additive disturbance, signal temporal logic, sequential optimization, mixed-integer programming  
% \end{IEEEkeywords}

\section{Introduction}
\label{sec:introduction}
Multi-agent systems (MASs) can be found in many applications, such as robotics, autonomous vehicles, 
and cyber-physical systems. 
When these systems are stochastic, the formal specification of system properties can be formulated 
in a probabilistic setting, enabling a trade-off between quantifying uncertainty and adjusting feasibility. 
As the complexity in control synthesis from temporal logic under uncertainty grows with the dimensionality 
of the overall system, existing approaches typically focus on single-agent schemes \cite{FarahaniTAC2019} 
or non-stochastic systems \cite{Liu2017,Buyukkocak2021}. 
%\lef{non-stochastic doesn't mean deterministic, right?} \commentPS{deterministic?} 

In this paper, we focus on signal temporal logic (STL) \cite{MalerSTL2004}, as a tool to formally 
formulate and verify specifications for a wide range of MAS applications. 
STL employs predicates as atomic elements, coupled with Boolean and temporal operators, 
allowing precise specification of complex spatio-temporal properties in a dynamical system. 
In a deterministic setting, it is possible to design sound and complete control algorithms 
that guarantee STL satisfaction \cite{MurrayCDC2014}, based on the quantitative semantics 
of STL \cite{Donze2010}. 
Here we consider stochastic MASs and a stochastic optimal control problem, where the goal 
is to satisfy a multi-agent STL specification with a predefined probability.
% while optimizing a performance criterion.

To address stochasticity in the STL framework, 
%recent works have focused on individual predicates using probability-based or risk-measure-based conditions.  
the works in \cite{Safaoui2020,LindemannCDC2020,LindemannTAC2022} propose risk constraints 
over predicates while preserving Boolean and temporal operators. 
Probabilistic STL in \cite{Sadigh2016} allows one to express uncertainty by incorporating 
random variables into predicates, while \cite{Sadigh2018} introduces chance-constrained 
temporal logic for modeling uncertainty in autonomous vehicles. 
Similar approaches are found in \cite{Li2017,Kyriakis2019}. 
Top-down approaches imposing chance constraints on the entire specification are explored 
in \cite{FarahaniTAC2019, ScherHSCC2022, Scher2022}. 
%Considering deterministic settings, robust techniques for STL synthesis have been 
%studied in \cite{Sadraddini2015,Farahani2015, RamanHSCC2015}.  
Although important, these works focus on single-agent or low-dimensional systems and 
lack guidance on extending to MASs. An extension of \cite{FarahaniTAC2019} to stochastic 
MASs under STL has recently been presented in \cite{Yang2023}, considering, 
however, only a single joint task per agent and bounded distributions.

Here, we propose a two-stage approach to solve a stochastic optimal control problem for 
discrete-time linear MASs subject to additive stochastic perturbations and a global STL 
specification permitting multiple individual and joint tasks per agent. First, we decompose 
the multi-agent dynamics into a deterministic system and an error closed-loop stochastic 
system, for which we construct a probabilistic reachable tube (PRT) \cite{HewingECC2018} 
as the Cartesian product of reachable sets of individual error systems. 
These are driven by stochastic disturbances within confidence regions (CRs) with a fixed 
probability. By assuming independence among individual disturbances, we show that the PRT 
probability can be controlled by the product of probabilities selected for each individual 
CR and a union-bound argument applied over time. Thus, by lower bounding the PRT probability 
by the specification probability, we can tighten all state constraints induced by the STL 
specification by solving tractable optimization problems over segments of the PRT. 
For multi-agent STL specifications, this is a less conservative alternative to tightening 
approaches relying on the STL structure \cite{FarahaniTAC2019,Yang2023}. 
An attainable feasible solution to the resulting deterministic problem can then be used to 
synthesize multi-agent trajectories that satisfy the STL specification with the desired 
confidence level. To the best of the authors' knowledge, this work is the first to address 
stochastic MASs under STL utilizing PRTs. Subsequently, to enhance scalability, we decompose 
the resulting deterministic problem into agent-level subproblems, which are solved 
iteratively according to a scheduling policy. We show that this iterative procedure is 
recursively feasible, ensures satisfaction of local tasks, and guarantees nondecreasing 
robustness for joint tasks.

The remainder of the paper is organized as follows. Preliminaries and the control problem 
setup are in Sec. \ref{sec:Prob_setup}. The construction of PRTs, the constraint tightening 
and the distributed control synthesis, are in Sec. \ref{sec:main_results}. 
An illustrative numerical example is in Sec. \ref{sec:example}, whereas concluding remarks 
are discussed in Sec. \ref{sec:concl}.

\section{Problem setup}\label{sec:Prob_setup}

\subsection{Notation and Preliminaries}

\noindent\textbf{Notation:} The sets of real numbers and nonnegative integers are $\R$ and $\N$, 
respectively. Let $N\in \N$. Then, $\N_{[0,N]}=\{0,1,\ldots,N\}$. 
The transpose of $\xi$ is $\xi^\intercal$. 
The identity matrix is $I_n\in\R^{n\times n}$. 
Let $x_{1},\ldots,x_{n}$ be vectors. 
Then, $x=(x_{1},\ldots,x_{n}) = [x_{1}^\intercal  \;\cdots\;x_{n}^\intercal  ]^\intercal$. 
We denote by $\bm{x}(a:b)=(x(a),\ldots,x(b))$ an aggregate vector consisting of $x(t)$, 
$t\in \N_{[a,b]}$, representing a trajectory. 
When it is clear from the context, we write $\bm{x}(t)$, omitting the endpoint. 
When $x(t)$, $t\in\N_{[a,b]}$, are random vectors, $\bm{x}(a:b)=(x(a),\ldots,x(b))$ 
is a random process. 
Let $x_i(t)$, for $t\in\N_{[0,N]}$ and $i\in\N_{[1,M]}$. 
Then, $\bm{x}(0:N)=(x(0),\ldots,x(N))$ denotes an aggregate trajectory when 
$x(t)=(x_1(t),\ldots,x_M(t))$, $t\in\N_{[0,N]}$. 
The remainder of the division of $a$ by $b$ is $\mathrm{mod}(a,b)$. 
A random variable (vector) $w$ following a distribution $\cD_w$ is denoted as 
$w \sim\cD_w$, the support of $\cD_w$ is $\mathrm{supp}(\cD_w)$, the expected value of 
$w$ is $\mathbb{E}(w)$, and the variance (covariance matrix) of $w$ is $\mathrm{Var}(w)$ 
($\mathrm{Cov}(w)$). 
The probability of event $Y$ is $\mathrm{Pr}\{Y\}$. 
The cardinality of a set $\cV$ is $|\cV|$. 
The Minkowski sum and the Pontryagin set difference of 
$S_1\subseteq \R^n$ and $S_2\subseteq \R^n$ are 
$S_1\oplus S_2=\{s_1+s_2 \;|\; s_1 \in S_1,\; s_2 \in S_2\}$ 
and $S_1\ominus S_2 = \{s_1\in S_1\;|\;s_1  +s_2 \in S_1, \forall s_2 \in S_2 \}$, respectively.

% \begin{lemma}
%     Let $A=\mathrm{diag}(A_1,\ldots,A_M)\in \R^{n\times m}$ be a block-diagonal matrix, where $A_i\in\R^{n_i\times m_i}$, $i\in\N_{[1,M]}$, $\sum_i n_i=n$, and $\sum_i m_i=m$, and $\cW=\cW_1\times \ldots\times \cW_M\subset \R^{n\times m}$, where $\cW_i\subset \R^{n_i\times m_i}$, $i\in\N_{[1,M]}$. Then, $A\cW=(A_1\cW_1)\times\ldots\times (A_M\cW_M)$.
% \end{lemma}
% \begin{proof}
    
% \end{proof}

\begin{lemma}[Distributivity of Minkowski sum]\label{lemma:Minkowski_sum}
    Let $\cX_i,\cY_i$ $\subseteq \R^{n_i}$, $i\in\N_{[1,M]}$. Then, $(\cX_1\times \cdots \times \cX_M) \oplus (\cY_1\times \cdots \times \cY_M)=(\cX_1\oplus \cY_1)\times \cdots\times (\cX_M\oplus \cY_M)$.
\end{lemma}
\begin{proof}
It holds that $(\cX_1\times \ldots \times \cX_M) \oplus (\cY_1\times \ldots \times \cY_M)=\{(x_1,\ldots,x_M)+(y_1,\ldots\,y_M)\;|\;\forall x_i\in\cX_i,\; i\in\N_{[1,M]}, \mathrm{ and }\; \forall y_i\in\cY_i,\; i\in\N_{[1,M]}\} = \{(x_1+y_1,\ldots,x_M+y_M)\;|\;\forall x_i\in\cX_i, y_i\in\cY_i,\; i\in\N_{[1,M]}\} = (\cX_1\oplus \cY_1)\times \ldots\times (\cX_M\oplus \cY_M)$.
\end{proof}

We consider STL formulas with standard syntax:
\begin{equation}\label{eq:STL_syntax}
    \varphi 
    {}\coloneqq{} 
    \top 
    {}\mid{}
    \pi
    {}\mid{}
    \lnot \phi  
    {}\mid{}
    \phi_1 \wedge \phi_2 
    {}\mid{}
    \phi_1 U_{[t_1,t_2]}\phi_2,
\end{equation}
where $\pi:=(\mu(x)\geq 0)$ is a predicate, $\mu(x):=a^\intercal  x+b$ is an affine predicate function, 
with $a\in\R^{n_x}$, $x\in\R^{n_x}$, and $b\in\R$, and $\phi$, $\phi_1$, and $\phi_2$ are STL formulas, 
which are built recursively using predicates $\pi$, logical operators $\neg$ and $\wedge$, and the 
\textit{until} temporal operator $U$, with $[t_1,t_2]\equiv \N_{[t_1,t_2]}$. 
We omit $\lor$ (\textit{or}), $\lozenge$ (\textit{eventually}) and $\square$ (\textit{always}) 
operators from \eqref{eq:STL_syntax} and the sequel, as these may be defined by \eqref{eq:STL_syntax}, 
e.g., $\phi_1\lor \phi_2 = \neg (\neg \phi_1 \wedge \neg \phi_2)$, 
$\lozenge_{[t_1,t_2]}\phi = \top U_{[t_1,t_2]}\phi$, 
and $\square_{[t_1,t_2]}\phi = \lnot \lozenge_{[t_1,t_2]}\lnot \phi$. 
%We are interested in specifications applied over a finite-time horizon, that is, $t_1$, $t_2$ above are finite.  

Let $\pi$ be a predicate and $\phi$ an STL formula. We write $\pi \in \phi$ to indicate that $\pi$ 
is part of the formula $\phi$. We denote by $\bm{x}(t) \models \phi$, $t\in\N$, the satisfaction 
of $\phi$ verified over $\bm{x}(t)=(x(t),x(t+1),\ldots)$. 
The validity of a formula $\phi$ can be verified using Boolean semantics:
$\bm{x}(t) \models \pi \Leftrightarrow \mu(x(t)) \geq 0$, 
$\bm{x}(t) \models \lnot \phi \Leftrightarrow \lnot(\bm{x}(t) \models \phi)$, 
$\bm{x}(t) \models \phi_1 \land \phi_2 \Leftrightarrow \bm{x}(t) \models \phi_1 \land \bm{x}(t) \models \phi_2$,  
    %$\bm{x}(t) \models \phi_1 \lor \phi_2 \Leftrightarrow \bm{x}(t) \models \phi_1 \lor \bm{x}(t) \models \phi_2$,  
    %$\bm{x}(t) \models \lozenge_{[a,b]}\phi \Leftrightarrow \exists \tau \in  [t+a,t+b],\; \bm{x}(\tau) \models \phi$, 
    %$\bm{x}(t) \models \square_{[a,b]}\phi \Leftrightarrow \forall \tau \in [t+a,t+b],\; \bm{x}(\tau) \models \phi$, 
$\bm{x}(t) \models \phi_1 U_{[a,b]}\phi_2 \Leftrightarrow \exists \tau\in  t\oplus \N_{[a,b]}$, 
s.t. $\bm{x}(\tau) \models \phi_2 \land \forall \tau'\in \N_{[t,\tau]}, \bm{x}(\tau') \models \phi_1$.

Based on the semantics introduced above, the horizon of a formula is recursively defined as 
\cite{MalerSTL2004}: $N^\pi = 0$, $N^{\lnot\phi} = N^{\phi}$, 
$N^{\phi_1\land \phi_2} = \max(N^{\phi_1}, N^{\phi_2})$, 
$N^{\phi_1\;U_{[a,b]}\phi_2} = b+ \max(N^{\phi_1}, N^{\phi_2})$.

STL is endowed with robustness metrics for assessing the robust satisfaction of a 
formula \cite{Donze2010}. A scalar-valued function $\rho^\phi: \R^n\times\cdots\times \R^n \to \R$ 
of a signal indicates how robustly a signal $\bm{x}(t)$ satisfies a formula $\phi$. 
The robustness function is defined recursively as follows:
$\rho^\pi(\bm{x}(t)) = \mu(x(t))$, $\rho^{\lnot \phi}(\bm{x}(t)) = -\rho^{\phi}(\bm{x}(t))$,\\ 
$\rho^{\phi_1 \land \phi_2}(\bm{x}(t)) = \min(\rho^{\phi_1}(\bm{x}(t)),\rho^{\phi_2}(\bm{x}(t)))$, \\
%$\rho^{\phi_1 \lor \phi_2}(\bm{x}(t)) = \max(\rho^{\phi_1}(\bm{x}(t)),\rho^{\phi_2}(\bm{x}(t)))$, \\
% $\rho^{\square_{[a,b]}\phi}(\bm{x}(t)) = \min_{\tau\in [t+a,t+b]} \rho^{\phi}(\bm{x}(\tau))$, \\
% $\rho^{\lozenge_{[a,b]}\phi}(\bm{x}(t)) = \max_{\tau\in [t+a,t+b]} \rho^{\phi}(\bm{x}(\tau))$,\\
${\tiny\rho^{\phi_1 U_{[a,b]}\phi_2}(\bm{x}(t)) = \max_{\tau\in t\oplus \N_{[a,b]}}\left( \min(Y_1(\tau),Y_2(\tau')) \right)}$, with $Y_1(\tau)=\rho^{\phi_1}(\bm{x}(\tau))$, $Y_2(\tau')=\min_{\tau'\in\N_{[t,\tau]}}\rho^{\phi_2}(\bm{x}(\tau'))$, $\pi$ being a predicate, and $\phi$, $\phi_1$, and $\phi_2$ being STL formulas.

% graph description
\begin{definition}\label{def:cliques}
    Let $\cG = (\cV,\cE)$ be an undirected graph containing no self-loops, with node set $\cV$, cardinality $M=|\cV|$, and edge set $\cE$. Let also $\cV'\subseteq \cV$, with $|\cV'|>1$, and define $\cE_{\cV'}\subseteq \cE$ as the set of edges attached to nodes $\cV'$. Then, $\cG'=(\cV',\cE_{\cV'})$ is a clique \cite{Orlin1977}, i.e., a complete subgraph of $\cG$, if $\cE_{\cV'}$ contains all possible edges between nodes $\cV'$. The set of cliques of $\cG$ is defined as 
    \begin{align}
        \cK = \{\nu\subseteq \cV: (\nu,\cE_{\nu}) \text{ is a complete subgraph of }\cG \}.
    \end{align}
\end{definition}
% \begin{remark}
%     According to Definition \ref{def:cliques}, the set of cliques of a connected, undirected graph $\cG = (\cV,\cE)$ with no self-loops contains all edges, called two-vertex cliques, and all complete subgraphs of $\cG$ with cardinality greater than two and less than or equal to $|\cV|$.  
% \end{remark}

Consider a graph $\cG = (\cV,\cE)$ with clique set $\cK$, a clique $\nu\in\cK$, with $\nu=(i_1,\ldots,i_{|\nu|})$, and vectors $x_{i_j}(t)$, $j\in\N_{[1,|\nu|]}$, with $t\in\N$, % and matrices $A_{i_j}$, $j\in\N_{[1,|\nu|]}$. 
Then, $x_\nu(t)=(x_{i_1}(t),\ldots,x_{i_{|\nu|}}(t))$ %and $A_\nu=\diag(A_{i_1},\ldots,A_{i_{|\nu|}})$ are
is an aggregate vector. % and an aggregate block-diagonal matrix, respectively. 
We denote by $\bm{x}_\nu(t)\models \phi_\nu$ the validity of an STL formula defined over the aggregate trajectory $\bm{x}_\nu(t)=(x_\nu(t),x_\nu(t+1),\ldots)$. If $\pi_\nu\in \phi_\nu$, $\pi_\nu:=(\mu_\nu(x_\nu)\geq 0)$, where $\mu_\nu(x_\nu)$ is an affine predicate function of $x_\nu$, with $x_\nu=(x_{i_1},\ldots,x_{i_{|\nu|}})$. 

\subsection{Multi-agent system}\label{sec:MAS}

% \subsection{Dynamics}
\subsubsection{Dynamics} We consider $M$ agents with dynamics
\begin{equation}\label{eq:individual_agent_dynamics}
    x_i(t+1)=A_ix_i(t) + B_iu_i(t) + w_i(t),
\end{equation}
where $x_i(t)\in\cX_i\subseteq \R^{n_i}$, $u_i\in\cU_i\subseteq\R^{m_i}$, and $w_i(t)\in\cW_i\subseteq \R^{n_i}$ 
are the state, input and disturbance vectors, respectively, the initial condition, $x_i(0)$, is known, 
$(A_i,B_i)$ is a stabilizable pair, with $A_i\in\R^{n_i\times n_i}$, $B_i\in\R^{n_i\times m_i}$, $i\in\N_{[1,M]}$, 
and $t\in\N$. By collecting individual state, input, and disturbance vectors, as 
$x(t)=(x_1(t),\ldots,x_M(t))\in\cX\subseteq \R^n$, 
$u(t)=(u_1(t),\ldots,u_M(t))\in\cU\subseteq \R^m$, 
and $w(t)=(w_1(t),\ldots,w_M(t))\in\cW\subseteq \R^n$, respectively, 
we write the dynamics of the entire MAS as
\begin{equation}\label{eq:MAS}
    x(t+1) = Ax(t)+Bu(t)+w(t),
\end{equation}
\sloppy
where $A=\diag(A_1, \ldots ,A_M)$, $B=\diag(B_1, \ldots, B_M)$, and the state, 
input, and disturbance sets are 
$\cX=\cX_1\times \cdots \times \cX_M$, 
$\cU=\cU_1\times \cdots \times \cU_M$, and 
$\cW=\cW_1\times \cdots \times \cW_M$, respectively. %We assume that \eqref{eq:MAS} is controllable, that is, the pairs $(A_i,B_i)$, $i\in\N_{[1,M]}$, are controllable. 

\subsubsection{Disturbance} We assume that the uncertain sequence $\bm{w}_i(0)=(w_i(0),w_i(1),\ldots)$, with $i\in\N_{[1,M]}$, is an independent and identically distributed random process, and that $w_i(t)$ is a random vector with mean $\mathbb{E}(w_i(t))=0$ and positive definite covariance matrix $\mathrm{Cov}(w_i(t))=Q_i$, which is known, for all $t\in\N$. We also assume that $w_i(t)$, $i\in\N_{[1,M]}$, are independent for $t\in\N$. We denote by $\cD_{w_i}$ the distribution of the disturbance $w_i(t)\in\cW_i$, where $\cW_i$ is its support, which may be unbounded. We also may write that $w(t)\sim\cD_w$, with $\mathrm{supp}(\cD_w)=\cW=\cW_1\times\cdots\times \cW_M$, $\mathbb{E}(w(t))=0$, $\mathrm{Cov}(w(t))=\diag(Q_1,\ldots,Q_M)=Q$.

Because the sequence $\bm{w}(0)=(w(0),w(1),\ldots,)$ is uncertain, the trajectory of the MAS \eqref{eq:MAS} becomes a random process, $\bm{x}(0)=(x(0),x(1),\ldots)$, where $x(t)$, $t\in\N$, are random vectors (except $x(0)=x_0$), each determined by the dynamics \eqref{eq:MAS} and admissible input and disturbance sequences $\bm{u}(0:t-1)$ and $\bm{w}(0:t-1)$, respectively, with $t\in\N_{\geq 1}$. To avoid confusion, we do not alter notation for different realizations of $\bm{x}(\cdot)$, $\bm{u}(\cdot)$, $\bm{w}(\cdot)$.

\subsubsection{STL specification} Let $\cV=\{1,\ldots,M\}$ be the set collecting the indices of all agents. 
The MAS is subject to a conjunctive STL formula $\phi$ with syntax as in \eqref{eq:STL_syntax}, where 
each conjunct is either a local subformula $\phi_i$ involving agent $i\in\cV$, or a joint subformula 
$\phi_\nu$ involving a subset of agents $\nu\subseteq\cV$. 
We call $\nu$ a clique. By collecting all cliques $\nu$ in $\cK_\phi$, the global STL task is 
% defined at the multi-agent level is
\begin{align}
    \phi = \bigwedge_{i\in \cV} \phi_i \wedge \bigwedge_{\nu\in \cK_\phi} \phi_\nu. \label{eq:global_phi}
\end{align}
The structure of $\phi$ in \eqref{eq:global_phi} induces an interaction graph $\cG = (\cV,\cE)$, 
where $\cV$ is the set of nodes, and 
$\cE=\{(\nu_i,\nu_j)\mid \nu_i,\nu_j\in\nu,\; i\neq j, \; \nu \in \cK_\phi\}$ 
is the set of edges. 
Let $\pi\coloneqq(\mu(y)\geq 0)$ be a predicate in $\phi$, where 
$\mu(y)=a^\intercal y+b$, with $a,y\in\R^{n_y}$ and $b\in\R$. 
The vector $y\in\R^{n_y}$ represents either an individual state vector, $x_i\in\R^{n_i}$, $i\in\N_{[1,M]}$, 
or an aggregate vector, $x_\nu\in\R^{n_\nu}$, collecting the states of agents in the clique 
$\nu\in\cK_\phi$. Formula $\phi$ can specify tasks between subsets of agents, by representing their entirety 
as a clique. For instance, let $\nu',\nu'' \subseteq \cV$ and consider formulas $\phi'$, $\phi''$, 
involving agents from $\nu'$, $\nu''$, respectively. 
Then, it is possible to specify $\phi_\nu=\phi'\ast\phi''$, with $\ast \in\{\wedge, \lor, U\}$, 
by defining the clique $\nu=\nu'\cup \nu''$. Without loss of generality, we assume the following.

% , $\cK$ is the set of cliques (cf. Definition \ref{def:cliques}), and $\cK_\phi\subseteq \cK$ is the  set of cliques that appear in $\phi$.\footnote{For example, let $\cG=(\cV,\cE)$ be an undirected graph, with $\cV=\{1,2,3\}$, $\cE=\{(1,2),(2,3),(1,3)\}$, and $\cK=\cE\cup\{(1,2,3)\}$. Then, $\cK_\phi=\cK$, $\cK_\phi=\cE$, or $\cK_\phi=\{(1,2,3)\}$.} The global STL task defined at multi-agent level is
% \begin{align}
%     \phi = \bigwedge_{i\in \cV} \phi_i \wedge \bigwedge_{\nu\in \cK_\phi} \phi_\nu, \label{eq:global_phi}
% \end{align}
% where $\phi_i$ indicates local subformulas involving individual agents with $|i|=1$, whereas $\phi_\nu$ indicates joint subformulas involving pairs of agents (represented by two-vertex cliques $\nu$, with $|\nu|=2$) or more than two agents (indicated by multi-vertex cliques $\nu$, with $2<|\nu|\leq M$). %This utilization of cliques allows us to capture interactions that go beyond individual or two-agent tasks, as typically represented in a vertex-edge graph structure. 

\begin{assumption}\label{ass:neg_pi}
    Let $\pi$ be a predicate. Then, if $\pi\in \phi$, $\neg\pi \notin \phi$.
\end{assumption}

% \begin{remark}
%     The structure of $\phi$ in \eqref{eq:global_phi} does not capture tasks among agents belonging to different cliques. If this is needed, one may construct a clique graph, $\cG'=(\cV',\cE')$, where a node $v'\in\cV'$ represents a clique $\nu\in\cK_\phi$, and $\epsilon'=(v',v'')\in\cE'$ represents a link between nodes $v',v''\in\cV'$. Then, \eqref{eq:global_phi} is extended to $\phi = (\bigwedge_{i\in \cV} \phi_i) \wedge (\bigwedge_{\nu\in \cK_\phi} \phi_\nu)\wedge (\bigwedge_{\epsilon'\in \cE'} \phi_{\epsilon'})$, allowing collaborative tasks assignments over more than one clique, e.g., $\phi_{\nu_1}\ast\phi_{\nu_2}$, where $\nu_1,\nu_2\in\cK_\phi$, and $\ast\in\{\wedge,\lor,U\}$. 
% \end{remark}

\subsection{Problem statement}

We wish to solve the  stochastic optimal control problem
\begin{subequations}\label{eq:multi_agent_problem}
\begin{align}
   &\operatorname*{Min.}_{\substack{\bm{u}(0),\\ \bm{x}(0)}}
        \mathbb{E}
        \left[
            \sum_{i=1}^M\big(\sum_{t=0}^{N-1}(\ell_i(x_i(t), u_i(t))) + V_{f,i}(x_i(N))\big)
        \right] \label{eq:cost_function_prob} \\
     &\mathrm{s.t.~}  x(t+1) = Ax(t)+Bu(t)+w(t),\; t\in\N_{[0,N)}, \label{eq:MAS_dynamics_prob} \\
        &\;\;\;\;\;\; \mathrm{Pr}\{\bm{x}(0) \models \phi\}\geq \theta, \; \mathrm{with}\; x(0)=x_0, \label{eq:MAS_STL_prob} 
        % &\;\;\;\;\;\;  \label{eq:MAS_initial_state_prob}
\end{align}    
\end{subequations}
where $\ell_i:\R^{n_i}\times \R^{m_i}\to \R$ and $V_{f,i}:\R^{n_i}\to \R$ are cost functions penalizing the first $N$ elements of the trajectories $\bm{x}_i(0)$, $\bm{u}_i(0)$, and the end point $x_i(N)$, respectively, with $i\in\N_{[1,M]}$, the multi-agent input trajectory $\bm{u}(0)$ is the optimization variable driving the trajectory $\bm{x}(0)=(x(0),\ldots,x(N))$, originating at $x_0$, to satisfy the formula $\phi$ with probability $\theta\in(0,1)$, as stated in \eqref{eq:MAS_STL_prob}, and $N$ is the underlying horizon of the formula. In fact, \eqref{eq:multi_agent_problem} is a planning problem, where minimization is taken over sequences of control actions, rather than over control policies.
Solving the problem directly is challenging due to the probabilistic constraint, the expectation operator in the cost function, and uncertain dynamics. To handle complexity, especially with large number of agents and complex $\phi$, we transform it into a deterministic problem, which subsequently, we decompose into smaller agent-level subproblems. Additionally, we make the following assumption.

% \begin{assumption}\label{ass:cost_functions}
%     Functions $l$, $l'$ are separable, that is, $l(x,u) = \sum_{i=1}^Ml_i(x_i,u_i)$ and $l'(x) = \sum_{i=1}^Ml_i'(x_i)$, with $l_i:\R^{n_i}\times \R^{m_i}\to \R$, $l_i':\R^{n_i}\to \R$. 
% \end{assumption}

\begin{assumption}\label{ass:global_problem}
% Given $\theta\in(0,1)$, Problem \eqref{eq:multi_agent_problem} is feasible: There exists admissible $\bm{u}(0)$ such that \eqref{eq:MAS_STL_prob} is satisfied, considering all realizations of $\bm{x}(0)$, each produced by an admissible realization of $\bm{w}(0)$, given $\bm{u}(0)$ and $x(0)=x_0$.  
For $x(0)=x_0$ and given $\theta\in(0,1)$, Problem \eqref{eq:multi_agent_problem} is feasible: 
There exists a sequence  $\bm{u}(0)$ such that the corresponding sequence of states, $\bm{x}(0)$, satisfies $\phi$
with probability at least $\theta$.
\end{assumption}

\begin{remark}
    %Despite Assumption \ref{ass:global_problem}, the actual multi-agent trajectory $\bm{x}(0)$ may not verify $\phi$ for some realizations of $\bm{w}(0)$. On the other hand, 
    Setting $\theta=1$ in \eqref{eq:MAS_STL_prob}, enforcing $\bm{x}(0)\models \phi$ $\forall \bm{w}(0)$, typically makes sense if $\mathrm{supp}(\cD_w)$ is bounded.
\end{remark}

\section{Main results}\label{sec:main_results}

\subsection{Error dynamics and construction of probabilistic tubes}
Due to the linearity in \eqref{eq:individual_agent_dynamics}, the state of each agent can be decomposed into a deterministic part, $z_i(t)$, and an error, $e_i(t)$, i.e., $x_i(t) = z_i(t) + e_i(t)$. 
Consider the causal control law $u_i(t)=K_ie_i(t)+v_i(t)$, where $K_i\in \R^{m_i\times n_i}$ 
is a stabilizing state-feedback gain for the pair $(A_i,B_i)$. Then, we may write
\begin{subequations}\label{eq:ith_decomposed_dynamics}
    \begin{align}  
    z_i(t+1)&=A_iz_i(t)+B_iv_i(t), \label{eq:ith_determ_dyn}\\
    e_i(t+1)&=\bar{A}_i e_i(t)+w_i(t),\label{eq:ith_error_dyn}
\end{align}    
\end{subequations}
where $z_i(0)=x_i(0)$, $e_i(0)=0$, and $\bar{A}_i=A_i+B_iK_i$. 
The above choice of $u_i(t)$ will allows us to control the size of the reachable sets of 
\eqref{eq:ith_error_dyn}, in light of the probabilistic constraint in \eqref{eq:MAS_STL_prob}. 
By defining the deterministic and stochastic components of 
$x(t)$, as $z(t)=(z_1(t),\ldots,z_M(t))$ and $e(t)=(e_1(t),\ldots,e_M(t))$, 
respectively, a deterministic input $v(t)=(v_1(t),\ldots,v_M(t))$, a block-diagonal state-feedback gain $K=\diag(K_1,\ldots,K_M)\in\R^{m\times n}$, and an aggregate block-diagonal closed-loop matrix $\bar{A}=\diag(\bar{A}_1,\ldots,\bar{A}_M)$, we decompose \eqref{eq:MAS} into
\begin{subequations}\label{eq:MAS_decomposed_dynamics}
    \begin{align}  
    z(t+1)&=Az(t)+Bv(t),\label{eq:MAS_determin_dyn}\\
    e(t+1)&=\bar{A}e(t)+w(t).\label{eq:MAS_error_dyn}
\end{align}    
\end{subequations}
Given a particular state feedback gain $K$, the error system \eqref{eq:MAS_error_dyn} 
can be analyzed independently of \eqref{eq:MAS_determin_dyn}. 
As a closed-loop system driven by the random vector $w(t)$, we can predict its trajectory 
$\bm{e}(0)=(e(0),\ldots,e(N))$, with $e(0)=0$, by calculating its reachable sets probabilistically. 
Probabilistic reachable sets and tubes for system \eqref{eq:MAS_error_dyn} are defined next.

\begin{definition}
    A set $E(t)\subseteq \R^n$, $t\in\N_{[0,N]}$, is called a $t$-step probabilistic reachable set ($t$-PRS) for \eqref{eq:MAS_error_dyn} at probability level $\hat{\theta}_t\in[0,1]$, if $\mathrm{Pr}\{e(t)\in E(t)\mid e(0)=0\}\geq \hat{\theta}_t$. 
\end{definition}
Since $e(0)=0$, it follows that $E(0)=\{0\}$. Also, it is worth noting that for a probability level $\hat{\theta}_t$ a $t$-PRS $E(t)$, $t\in\N_{[0,N]}$, for \eqref{eq:MAS_error_dyn} is not unique. 
\begin{definition}\label{def:tubes}
    Let $\bm{e}(0)=(e(0),\ldots,e(N))$ be a trajectory of \eqref{eq:MAS_error_dyn}. Then, $\bm{E}\subseteq\R^n\times \cdots\times \R^n$ is called a probabilistic reachable tube (PRT) for \eqref{eq:MAS_error_dyn} at probability level $\Theta\in[0,1]$, if $\mathrm{Pr}\{\bm{e}(0)\in \bm{E}\}\geq \Theta$. 
\end{definition}

Next, we show how to obtain $t$-PRSs for \eqref{eq:MAS_error_dyn} at probability levels 
$\hat{\theta}_t$, $t\in\N_{[0,N]}$, by identifying subsets of $\cW$ with a degree of
concentration of $\cD_{w}$ \cite{HewingECC2018}.
\begin{definition}
    Let $w\sim \cD_w$, with $\mathrm{supp}(\cD_w)=\cW$. We call $\mathscr{E}_\theta(\cD_w)\subseteq \cW$ 
    a confidence region (CR) for $w\in\cW$ at probability level $\theta$, 
    if $\mathrm{Pr}\{w\in \mathscr{E}_\theta(\cD_w)\}\geq \theta$.
\end{definition}

CRs for $w_i(t)$, $i\in\N_{[1,M]}$, can be approximated via Monte Carlo methods or computed 
analytically using concentration inequalities depending on the properties of $\cD_{w}$. 
Here, since we only know that $\mathbb{E}(w_i(t))=0$ and $\mathrm{Cov}(w_i(t))=Q_i>0$, 
for $t\in\N$ and $i\in\N_{[1,M]}$, we construct ellipsoidal CRs at probability level 
$\theta_i$ as $\mathscr{E}_{\theta_i}(\cD_{w_i})=\{w_i\in \R^{n_i}\;|\;w_i^\intercal  Q_i^{-1}w_i\leq n_i/\theta_i\}$, 
by the multivariate Chebyshev's inequality. 
Next, we show how to construct a CR for the aggregate random vector $w(t)=(w_1(t),\ldots,w_M(t))$.
\begin{lemma}\label{lemma:conf_region}
    Let $\mathscr{E}_{\theta_i}(\cD_{w_i})$ be a CR for $w_i(t)\in\cW_i$, $i\in\N_{[1,M]}$, 
    at probability level $\theta_i$, for all $t\in\N$. 
    Then, 
    \(\mathscr{E}_{\hat{\theta}}(\cD_w)
    {}\supseteq{}
    \mathscr{E}_{\theta_1}(\cD_{w_1})\times \cdots\times \mathscr{E}_{\theta_M}(\cD_{w_M})\), 
    is a CR of $w(t)\in\cW$ at probability level $\hat{\theta}$ for all $t\in\N$, 
    where $\hat{\theta}\geq \Pi_{i=1}^M\theta_i$.  
\end{lemma}
\begin{proof}
Without loss of generality let $M=2$. Then, 
\(\mathrm{Pr}\{w(t)\in\mathscr{E}_{\hat{\theta}}(\cD_w)\}
{}\geq{}
\mathrm{Pr}\{(w_1(t)\in\mathscr{E}_{\theta_1}(\cD_{w_1}))\wedge(w_2(t)\in\mathscr{E}_{\theta_2}(\cD_{w_2}))\}
{}={}
\mathrm{Pr}\{ w_1(t)\in\mathscr{E}_{\theta_1}(\cD_{w_1}) \}
\mathrm{Pr}\{w_2(t)\in\mathscr{E}_{\theta_2}(\cD_{w_2})\}
{}\geq{} 
\theta_1 \theta_2
\), 
which is true due to independence of $w_1(t)$, $w_2(t)$, for all $t\in\N$.
% \begin{multline*}
%     \mathrm{Pr}\{
%     (w_1(t)\in\mathscr{E}_{\theta_1}(\cD_{w_1}))
%     \wedge
%     (w_2(t)\in\mathscr{E}_{\theta_2}(\cD_{w_2}))
%     \}
%     \\
%     {}={}
%     ?
% \end{multline*}
%     $\mathrm{Pr}\{w(t)\in\mathscr{E}_{\hat{\theta}}(\cD_w)\}\geq\mathrm{Pr}\{(w_1(t)\in\mathscr{E}_{\theta_1}(\cD_{w_1}))\wedge\cdots\wedge(w_M(t)\in\mathscr{E}_{\theta_M}(\cD_{w_M}))\}=1-\mathrm{Pr}\{(w_1(t)\notin\mathscr{E}_{\theta_1}(\cD_{w_1}))\lor\cdots\lor(w_M(t)\notin\mathscr{E}_{\theta_M}(\cD_{w_M}))\}\geq 1-\sum_{i=1}^M\mathrm{Pr}\{w_i(t)\notin\mathscr{E}_{\theta_i}(\cD_{w_i})\}=1-\sum_{i=1}^M(1-\theta_i)$, which follows from Boole's Inequality (BI). 
\end{proof}

Based on the CR construction of the disturbance $w(t)$, we construct $t$-PRSs at 
certain probability levels for the multi-agent error system \eqref{eq:MAS_error_dyn} as follows. 

\begin{proposition}\label{prop:tPRS_theta}
    Let 
    \(
    \mathscr{E}_{\hat{\theta}}(\cD_{w})
    {}={}
    \mathscr{E}_{\theta_1}(\cD_{w_1})\times \cdots\times \mathscr{E}_{\theta_M}(\cD_{w_M})
    \)
    be a CR for $w(t)$, where $\mathscr{E}_{\theta_i}(\cD_{w_i})$ is an ellipsoidal CR for 
    $w_i(t)$ at probability level $\theta_i$, with $i\in\N_{[1,M]}$. 
    Then, the sets $E(t)\subseteq\R^n$, $t\in\N_{[0,N]}$, which are recursively defined as 
    $E(t+1)=\bar{A}E(t)\oplus \mathscr{E}_{\hat{\theta}}(\cD_{w})$, with 
    $E(0)=\{0\}\times \cdots\times \{0\}$, 
    are $t$-PRSs for \eqref{eq:MAS_error_dyn} at probability level $\hat{\theta}\geq \Pi_{i=1}^M\theta_i$, and $E(t)=E_1(t)\times \cdots\times E_M(t)$, $t\in\N_{[0,N]}$, where $E_i(t)$ is a $t$-PRS for \eqref{eq:ith_error_dyn} at probability level $\theta_i$, with $i\in\N_{[1,M]}$.  
\end{proposition}
\begin{proof}
    Since $E(0)=\{0\}\times \cdots\times \{0\}$, we may write 
    $E(0)=E_1(0)\times \cdots\times E_M(0)$, 
    with $E_i(0)=\{0\}$, $i\in\N_{[1,M]}$, 
    from which we compute 
    \(
        E(1)
        {}={}
        \bar{A}E(0)\oplus \mathscr{E}_{\hat{\theta}}(\cD_w)
        {}={}
        (\diag(\bar{A}_1,\ldots,\bar{A}_M)E_1(0)\times \cdots\times E_M(0))
        {}\oplus{}
        (\mathscr{E}_{\theta_1}(\cD_{w_1})\times \cdots\times \mathscr{E}_{\theta_M}(\cD_{w_M}))
        {}={}
        (\bar{A}_1E_1(0)\times \cdots\times \bar{A}_ME_M(0))
        {}\oplus{}
        (\mathscr{E}_{\theta_1}(\cD_{w_1})\times \cdots\times \mathscr{E}_{\theta_M}(\cD_{w_M}))
    \), 
    which from Lemma \ref{lemma:Minkowski_sum} results in 
    \(
    E(1)
    {}={}
    \bar{A}_1 E_1(0)
    {}\oplus{}
    \mathscr{E}_{\theta_1}(\cD_{w_1})\times \cdots\times \bar{A}_M E_M(0)
    {}\oplus{}
    \mathscr{E}_{\theta_M}(\cD_{w_M})
    \), 
    that is, $E(1)=E_1(1)\times \cdots\times E_M(1)$, 
    with $E_i(1)=\bar{A}_i E_i(0)\oplus\mathscr{E}_{\theta_i}(\cD_{w_i})$, $i\in\N_{[1,M]}$. 
    Following the recursion, one can show that 
    \(E_i(t+1)=\bar{A}_iE_i(t)\oplus \mathscr{E}_{\theta_i}(\cD_{w_i})\), 
    for \(t\in\N_{[0,N-1]}\), and 
    \(E(t)=E_1(t)\times \cdots\times E_M(t)\)
    for $t\in\N_{[0,N]}$.  

    Let now $\cD_{e_i(t)}$ be the distribution of $e_i(t)$, and $\mathscr{E}_{\theta_i}(\cD_{e_i(t)})$ be a CR for $e_i(t)$ at probability $\theta_i$. Since, $\mathscr{E}_{\theta_i}(\cD_{e_i(0)})\subseteq E_i(0)=\{0\}$, we have $\bar{A}_i\mathscr{E}_{\theta_i}(\cD_{e(0)})\oplus \mathscr{E}_{\theta_i}(\cD_{w_i})\subseteq \bar{A}_iE_i(0)\oplus \mathscr{E}_{\theta_i}(\cD_{w_i})=E_i(1)$, so $\mathscr{E}_{\theta_i}(\cD_{e_i(1)})\subseteq E_i(1)$, as $\mathscr{E}_{\theta_i}(\cD_{e_i(t+1)})\subseteq \bar{A}_i\mathscr{E}_{\theta_i}(\cD_{e_i(t)})\oplus\mathscr{E}_{\theta_i}(\cD_{w_i})$ for all $t\in\N$, since $\mathscr{E}_{\theta_i}(\cD_{w_i})$ is an ellipsoidal region by \cite[Cor. 4]{HewingECC2018}. Inductively we show that $\mathscr{E}_{\theta_i}(\cD_{e_i(t)})\subseteq E_i(t)$, $i\in\N_{[1,M]}$, for all $t\in\N$. The latter implies that $\mathrm{Pr}\{e_i(t)\in E_i(t)\}\geq \theta_i$, $i\in\N_{[1,M]}$, from which we have $\mathrm{Pr}\{e(t)\in E(t)\}=\mathrm{Pr}\{(e_1(t)\in E_1(t))\wedge\cdots\wedge(e_M(t)\in E_M(t))= \mathrm{Pr}\{(e_1(t)\in E_1(t))\} \cdots  \mathrm{Pr}\{(e_M(t)\in E_M(t))\}\geq \theta_1\theta_2\cdots\theta_M$, which follows from the independence of $E_i(t)$, $i\in\N_{[1,M]}$, by construction, for all $t\in\N$.
\end{proof}

% From Definition \ref{def:tubes} and  Proposition \ref{prop:tPRS_theta}, we may synthesize a PRT for \eqref{eq:MAS_error_dyn} of a desired probability level by the $t$-PRSs of \eqref{eq:MAS_error_dyn}. This is summarized next.
Proposition \ref{prop:tPRS_theta} leads to the following PRT result.
\begin{theorem}\label{thm:tubes}
    Let $\mathscr{E}_{\hat{\theta}}(\cD_{w})=\mathscr{E}_{\theta_1}(\cD_{w_1})\times \cdots\times \mathscr{E}_{\theta_M}(\cD_{w_M})$ be a CR for $w(t)$, and define $t$-PRSs, $E(t)$, $t\in\N_{[0,N]}$ for system \eqref{eq:MAS_error_dyn} at probability level $\hat{\theta}\geq \Pi_{i=1}^M\theta_i$, where $\theta_i$ is the confidence level of the region $\mathscr{E}_{\theta_i}(\cD_{w_i})$, $i\in\N_{[1,M]}$, as in Proposition \ref{prop:tPRS_theta}. Then, i) $\bm{E}_i=E_i(0)\times \cdots\times E_i(N)$ is a PRT for \eqref{eq:ith_error_dyn} at probability level $\Theta_i\geq 1-N(1-\theta_i)$, where $E_i(t)$, $t\in\N_{[0,N]}$, is a $t$-PRS for \eqref{eq:ith_error_dyn} at probability level $\theta_i$. ii) $\bm{E}=E(0)\times \cdots \times E(N)$ is a PRT for \eqref{eq:MAS_error_dyn} at probability level $\Theta =\Pi_{i=1}^M\Theta_i$. iii) Let $e_\nu(t+1)=\bar{A}_\nu e_\nu(t)+w_\nu(t)$ be the aggregate system collecting individual error systems from the clique $\nu\in\cK_\phi$, where $\nu=(i_1,\ldots,i_{|\nu|})$, $e_\nu(t)=(e_{i_1}(t),\ldots,e_{i_{|\nu|}}(t))$, $w_\nu(t)=(w_{i_1}(t),\ldots,w_{i_{|\nu|}}(t))$, and $\bar{A}_\nu=\diag(\bar{A}_{i_1},\ldots,\bar{A}_{i_{|\nu|}})$, and let $E_\nu(t)=E_{i_1}(t)\times \cdots \times E_{i_{|\nu|}}(t)$, $t\in\N_{[0,N]}$,  be its $t$-PRSs, with $E_{i_j}(t)$, being $t$-PRS for \eqref{eq:ith_error_dyn} at probability level $\theta_{i_j}$, with $j\in\N_{[1,|\nu|]}$. Then, $\bm{E}_\nu=E_\nu(0)\times \cdots\times E_\nu(N)$ is a PRT at probability level $\Theta =\Pi_{j=1}^{|\nu|}\Theta_{i_j}$.
\end{theorem}
\begin{proof}
    i) From Proposition \ref{prop:tPRS_theta}, we have $E(t)=E_1(t)\times \cdots \times E_M(t)$, where $E_i(t)$ is a $t$-PRS for \eqref{eq:ith_error_dyn} at probability level $\theta_i$. Let $\bm{e}_i(0)=(e_i(0),\ldots,e_i(N))$ be a trajectory of \eqref{eq:ith_error_dyn}. Then, $\mathrm{Pr}\{\bm{e}_i(0)\in\bm{E}_i\}=\mathrm{Pr}\{(e_i(0)\in E_i(0))\wedge\cdots\wedge(e_i(N)\in E_i(N))\}=1-\mathrm{Pr}\{(e_i(0)\notin E_i(0))\lor\cdots\lor(e_i(N)\notin E_i(N))\}\geq 
    1-\sum_{t=0}^N\mathrm{Pr}\{e_i(t)\notin E_i(t)\}=1-N(1-\theta_i)$, where we used Boole's inequality (BI), that $E_i(t)$, $t\in\N_{[0,N]}$, is a $t$-PRS at probability level $\theta_i$, and $\mathrm{Pr}(e_i(0)\notin E_i(0)\}=0$. ii) $\mathrm{Pr}\{\bm{e}(0)\in\bm{E}\}=\mathrm{Pr}\{\left((e_1(0)\in E_1(0))\wedge\cdots\wedge(e_M(0)\in E_M(0))\right)\wedge\cdots\wedge\left((e_1(N)\in E_1(N))\wedge\cdots\wedge(e_M(N)\in E_M(N))\right)\}=\mathrm{Pr}\{\left((e_1(0)\in E_1(0))\wedge\cdots\wedge(e_1(N)\in E_1(N))\right)\wedge\cdots\wedge\left((e_M(0)\in E_M(0))\wedge\cdots\wedge(e_M(N)\in E_M(N))\right)\}=\mathrm{Pr}\{(\bm{e}_1(0)\in \bm{E}_1)\wedge\cdots\wedge(\bm{e}_M(0)\in \bm{E}_M)\} =\mathrm{Pr}\{\bm{e}_1(0)\in \bm{E}_1\}\cdots\mathrm{Pr}\{\bm{e}_M(0)\in \bm{E}_M\}=\Pi_{i=1}^M\Theta_i$, which follows from the independence of the PRTs $\bm{E}_i$, $i\in\N_{[1,M]}$, by construction. iii) By setting $M=|\nu|$ the result follows from Proposition \ref{prop:tPRS_theta} and item ii) herein. 
\end{proof}

\begin{remark}
    Note that our PRT construction reduces conservatism for a large number of agents, while utilizing the union-bound argument only over time.  
    This may require conservative choices for the probability levels, $\theta_i$, 
    for the CRs of $w_i(t)$, $i\in\N_{[1,M]}$, for large horizons. 
    To construct, e.g., a PRT for \eqref{eq:MAS_error_dyn} at probability level 
    $\Theta$, one may select uniform probability levels for the CRs as 
    $\theta_i\geq 1-\frac{1-\Theta^{\frac{1}{M}}}{N}$,
    where $\theta_i\to 1$ for large $N$, regardless of $\Theta$. 
\end{remark}

% \begin{corollary}
%     Let $\nu\in\cK_\phi$ be a clique of $\cG=(\cV,\cE)$, with $\nu=(i_1,\ldots,i_{|\nu|})$, and $x_\nu(t+1)=A_\nu x_\nu(t)+B_\nu u_\nu(t)+w_\nu(t)$, where $x_\nu(t)=(x_{i_1}(t),\ldots,x_{i_{|\nu|}}(t))$, $u_\nu(t)=(u_{i_1}(t),\ldots,u_{i_{|\nu|}}(t))$, $w_\nu(t)=(w_{i_1}(t),\ldots,w_{i_{|\nu|}}(t))$ $A_\nu=\mathrm{diag}(A_{i_1},\ldots,A_{i_{|\nu|}})$, and $B_\nu=\mathrm{diag}(B_{i_1},\ldots,B_{i_{|\nu|}})$, be the aggregate system collecting the dynamics of the agents in clique $\nu$, which is decomposed into deterministic and stochastic components as 
%     \begin{subequations}\label{eq:nu_decomposed_dynamics}
%     \begin{align}  
%     z(t+1)&=Az(t)+Bv(t),\label{eq:nu_determin_dyn}\\
%     e(t+1)&=\bar{A}e(t)+w(t),\label{eq:nu_error_dyn}
% \end{align}    
% \end{subequations}
%     where $z_\nu(t)=(z_{i_1}(t),\ldots,z_{i_{|\nu|}}(t))$, $e_\nu(t)=(e_{i_1}(t),\ldots,e_{i_{|\nu|}}(t))$, $v_\nu(t)=(v_{i_1}(t),\ldots,v_{i_{|\nu|}}(t))$, and $\bar{A}_\nu=\mathrm{diag}(\bar{A}_{i_1},\ldots,\bar{A}_{i_{|\nu|}})$.
    
% \end{corollary}

\subsection{Constraint tightening}

Based on the decomposition of \eqref{eq:MAS} into deterministic and stochastic components, we aim to design a trajectory for the deterministic system \eqref{eq:MAS_determin_dyn} that 
satisfies an STL formula derived from $\phi$, incorporating \textit{tighter} predicates. %, the truth values of which are not violated provided that the trajectory of the error system \eqref{eq:MAS_error_dyn} remains within a PRT at probability level $\theta$. 
The following proposition underpins this approach. 
\begin{proposition}\label{prop:MAS_chance_constraint}
Let $\bm{x}(0)=(x(0),\ldots,x(N))$ be a trajectory of \eqref{eq:MAS_dynamics_prob}, and $\bm{z}(0)=(z(0),\ldots,z(N))$ and $\bm{e}(0)=(e(0),\ldots,e(N))$ be its deterministic and stochastic components with respect to the dynamics \eqref{eq:MAS_determin_dyn} and \eqref{eq:MAS_error_dyn}, respectively, that is, $\bm{x}(0)=\bm{z}(0)+\bm{e}(0)$. Suppose that $\mathrm{Pr}\{\bm{e}(0)\in\bm{E}\}\geq \theta$, for some $\bm{E}=E(0)\times\cdots\times E(N)$, with $E(t)\subseteq \R^n$, $t\in\N$. If $\bm{z}(0)+\bm{e}(0)\models \phi$ for all $\bm{e}(0)\in\bm{E}$, then $\mathrm{Pr}\{\bm{x}(0)\models \phi\}\geq \theta$.
\end{proposition}
\begin{proof}
    Define events $Y_x:=\bm{x}(0)\models \phi$, $Y_e:=\bm{e}(0)\in\bm{E}$, and $Y_e':=\bm{e}(0)\notin \bm{E}$. From the law of total probability, we have $\mathrm{Pr}\{Y_x\}=\mathrm{Pr}\{Y_x|Y_e\}\mathrm{Pr}\{Y_e\}+\mathrm{Pr}\{Y_x|Y_e'\}\mathrm{Pr}\{Y_e'\}\geq \theta$, since by assumption, $\mathrm{Pr}\{Y_x|Y_e\}=1$ and $\mathrm{Pr}\{Y_e\}\geq \theta$, and $\mathrm{Pr}\{Y_x|Y_e'\}\mathrm{Pr}\{Y_e'\}\geq 0$.  
\end{proof}

Let $\bm{E}$ be a PRT for \eqref{eq:MAS_error_dyn} at probability level $\theta$. Next, we construct a formula $\psi$ such that $\bm{z}(0)\models \psi$ implies $\bm{z}(0)+\bm{e}(0)\models \phi$, for all $\bm{e}(0)\in\bm{E}$, that is, $\mathrm{Pr}\{\bm{x}(0)\models \phi\} \geq \theta$, by Proposition \ref{prop:MAS_chance_constraint}. Formula $\psi$ has identical Boolean and temporal operators with $\phi$ in \eqref{eq:global_phi}, and retains its multi-agent structure:
\begin{align}
    \psi = \bigwedge_{i\in \cV} \psi_i \wedge \bigwedge_{\nu\in \cK_\phi} \psi_\nu. \label{eq:global_psi}
\end{align}
Let $\pi:=(a^\intercal y+b\geq 0)$ be a predicate, with $a,y\in\R^{n_y}$, $b\in\R$. We denote by $\tau(\pi)$ the \textit{tighter} version of $\pi$, where $\tau(\pi)\in \psi$ if $\pi\in \phi$, and $\neg \tau(\pi)\in \psi$ if $\neg \pi \in \phi$, with 
\begin{subequations}\label{eq:all_predicates}
\begin{align}
   &\tau(\pi):=(a^\intercal y + b +\min_{g\in G} a^\intercal g \geq 0),\; \text{if\;} \tau(\pi) \in \psi, \label{eq:tight1}\\
    &\tau(\pi):=(a^\intercal y + b +\max_{g\in G} a^\intercal g \geq 0), \text{if\;}\neg\tau(\pi)\in \psi. \label{eq:tight2}
\end{align}    
\end{subequations}
Here, $G=\bigcup_{t=1}^N E_i(t)$ if $y=x_i(t)$, for some $i\in \cV$,
or $G=\bigcup_{t=1}^N E_\nu(t)$ if $y=x_\nu(t)$ for some $\nu \in \cK_\phi$. Also, $a$, $b$ in \eqref{eq:all_predicates} indicate the parameters of predicate functions associated with predicates contained in $\phi_i$ or $\phi_\nu$ (retained in $\psi_i$ or $\psi_\nu$, as per \eqref{eq:all_predicates}), with $i\in\cV$, $\nu\in\cK_\phi$. 
%For instance, if $y=x_i(t)$ ($y=x_\nu(t)$), $\pi'=\top$ if $a^\intercal x_i(t) + b +\min_{g\in E_i(t)} a^\intercal g \geq 0$ ($a^\intercal x_\nu(t) + b +\min_{g\in E_\nu(t)} a^\intercal g \geq 0$), and $\pi''=\top$ if $a^\intercal x_i(t) + b +\max_{g\in E_i(t)} a^\intercal g \geq 0$ ($a^\intercal x_\nu(t) + b +\max_{g\in E_\nu(t)} a^\intercal g \geq 0$). 
Distinguishing tightening between \eqref{eq:tight1} or \eqref{eq:tight2} while constructing $\psi$ is possible by Assumption \ref{ass:neg_pi}, which is not restrictive since we can always have $\pi, \neg\pi'\in\phi$, where $\pi'=\pi$. We also remark that the optimizations in \eqref{eq:all_predicates} are tractable since the domain $G$ is the union of finitely many convex and compact sets by the construction of $t$-PRSs, $E_i(t)$, $E_\nu(t)$, $t\in\N_{[1,N]}$, in Proposition \ref{prop:tPRS_theta} and Theorem \ref{thm:tubes}. Practically, the tightening in \eqref{eq:all_predicates} can be retrieved by solving a convex optimization by taking the convex hull of $G$, or, better, by solving $N$ convex optimization problems, one for every $E_i(t)$, $E_\nu(t)$, $t\in\N_{[1,N]}$, and selecting the worst-case (minimum for \eqref{eq:tight1} and maximum for \eqref{eq:tight2}) solution among them. We are now ready to state the following result.

\begin{theorem}\label{thm:deterministic_problem}
    Let $\psi$ be the STL formula resulting from $\phi$ according to \eqref{eq:global_psi}-\eqref{eq:all_predicates}, and assume that the deterministic problem:
    \begin{subequations}\label{eq:multi_agent_problem_tight}
    \begin{align}
    &\operatorname*{Minim.}_{\substack{\bm{v}(0),\\ \bm{z}(0)}} \sum_{i=1}^M\left(\sum_{t=0}^{N-1}(\ell_i(z_i(t), v_i(t))) + V_{f,i}(z_i(N))\right)
         \label{eq:cost_function_prob_tight} \\
     &\mathrm{s.t.~}  z(t+1) = Az(t)+Bv(t),\; t\in\N_{[0,N)},\label{eq:MAS_dynamics_prob_deterministic} \\
        &\;\;\;\;\;\; \bm{z}(0) \models \psi, \; \mathrm{with}\; z(0)=x_0, \label{eq:MAS_STL_prob_tight} 
        % &\;\;\;\;\;\; z(0)=x_0 \label{eq:MAS_initial_state_prob_tight}
    \end{align}    
    \end{subequations}
    has a feasible solution $\bm{v}(0)=(v(0),\ldots,v(N-1))$, with $v(t)\in \cU\ominus KE(t)$, $t\in\N_{[0,N-1]}$, where $K$ is a stabilising gain for $(A, B)$ in \eqref{eq:MAS_error_dyn}, and $E(t)=E_1(t)\times \cdots\times E_M(t)$, $E_i(t)$, $i\in\N_{[1,M]}$, being $t$-PRS for \eqref{eq:ith_error_dyn}, at probability level $\theta_i$, such that $\Pi_{i=1}^M\left(1-N(1-\theta_i)\right)\geq \theta$. Let $\bm{e}(0)=(e(0),\ldots,e(N))$ be a trajectory of \eqref{eq:MAS_error_dyn}. Then, $\bm{u}(0)=\tilde{K}\bm{e}(0)+\bm{v}(0)$ is a feasible solution for \eqref{eq:multi_agent_problem}, where $\tilde{K}=\diag(K,\ldots,K)$.
\end{theorem}
\begin{proof}
    By Theorem \ref{thm:tubes}, $\bm{E}=E(0)\times \cdots\times E(N)$ is a PRT for \eqref{eq:MAS_error_dyn} at probability level $\Theta\geq \theta$, that is, $\mathrm{Pr}\{\bm{e}(0)\in\bm{E}\}\geq \theta$. Let $\bm{z}(0)$ be a trajectory resulting from the input trajectory $\bm{v}(0)$ starting from $z(0)=x_0$, satisfying \eqref{eq:MAS_STL_prob_tight}, that is, $\bm{z}(0)\models \psi$. By the tightening in \eqref{eq:all_predicates}, we have that for all $\pi\in\phi$, $\tau(\pi)\in\psi$, such that if $\bm{z}(0)\models \tau(\pi)$, then $\bm{z}(0)+\bm{e}(0)\models \pi$ for all $\bm{e}(0)\in\bm{E}$, and for all $\neg\pi'\in\phi$, $\neg\tau(\pi')\in\psi$, such that if $\bm{z}(0)\models \neg\tau(\pi')$, then $\bm{z}(0)+\bm{e}(0)\models \neg\pi'$ for all $\bm{e}(0)\in\bm{E}$. Since $\phi$ and $\psi$ differ only in predicates in the aforementioned manner, it follows that if $\bm{z}(0)\models \psi$, then $\bm{z}(0)+\bm{e}(0)\models\phi$ for all $\bm{e}(0)\in\bm{E}$. Since $\bm{u}(0)=\tilde{K}\bm{e}(0)+\bm{v}(0)$ is a feasible input trajectory for \eqref{eq:MAS} for all $\bm{e}(0)\in \bm{E}$, the resulting state trajectory of \eqref{eq:MAS}, $\bm{x}(0)=\bm{z}(0)+\bm{e}(0)$, ensures that $\mathrm{Pr}\{\bm{x}(0)\models \phi\}=1$ for all $\bm{e}(0)\in \bm{E}$. Then, the result follows by Proposition \ref{prop:MAS_chance_constraint} since $\mathrm{Pr}\{\bm{e}(0)\in \bm{E}\}\geq \theta$.
\end{proof}

Clearly, apart from the volume of $\bm{E}$, the gain $K$ affects the feasible domain of \eqref{eq:multi_agent_problem_tight} too. Due to space limitations, we will present its construction in future work. By selecting $\|{}\cdot{}\|_1$-based costs, problem \eqref{eq:multi_agent_problem_tight} can be formulated as a mixed-integer linear program (MILP) using binary encoding for STL constraints \cite{MurrayCDC2014}. The complexity of \eqref{eq:multi_agent_problem_tight}, which grows with the dimension of the MAS, is addressed in the next section.

\subsection{Distributed control synthesis}\label{sec:distributed_control_synthesis}

We propose an iterative procedure that handles the complexity of \eqref{eq:multi_agent_problem_tight}, extending the practicality of the multi-agent framework at hand. First, we assume the following.
\begin{assumption}\label{ass:global_problem_tight}
Let the $t$-PRSs, $E(t)$, $t\in\N_{[0,N]}$, for \eqref{eq:MAS_error_dyn}, and the stabilizing gain, $K$, for the pair $(A,B)$, be as constructed in Theorem \ref{thm:deterministic_problem}. Then, \eqref{eq:multi_agent_problem_tight} has a feasible solution $\bm{v}(0)=(v(0),\ldots,v(N-1))$, where $v(t)\in \cU\ominus KE(t)$, $t\in\N_{[0,N-1]}$.
\end{assumption}

\subsubsection{Decomposition of STL formula $\psi$} For a node $i$ participating in at least one clique, i.e., $i\in\nu$, with $\nu\in\cK_\phi$, we define $\cT_i$ by the set of cliques containing $i$ excluding $i$, i.e.,
\begin{equation}\label{eq:Ti}
    \cT_i = \{\nu\setminus i: \nu \in \operatorname{cl}(i)\},
\end{equation}
where $\operatorname{cl}(i) = \{\nu \in \cK_\phi,\; \nu \ni i\}$ is the set of cliques that contain $i$. Let $j\in \cT_i$, with $j=(i_1,\ldots,i_{|j|})$. Let a trajectory $\bm{z}_{ij}(0)=(z_{ij}(0),\ldots,z_{ij}(N))$, where $z_{ij}(t)=(z_{i_1}(t),\ldots,z_i(t),\ldots,z_{i_{|j|}}(t))$, with $t\in\N_{[0,N]}$, and the order $i_1<\ldots<i<\ldots<i_{|j|}$ being specified by the lexicographic ordering of the node set $\cV=\N_{[1,M]}$. Using \eqref{eq:Ti}, an  equivalent formula to the \textit{tighter} formula  \eqref{eq:global_psi}-\eqref{eq:all_predicates}, $\psi$, is defined as $\hat{\psi}=\bigwedge_{i\in \cV}\hat{\psi}_i$, where
\begin{equation}\label{eq:psi_hat}
\hat{\psi}_i=\psi_i\wedge\bigwedge_{j\in\cT_i}\psi_{ij}.
\end{equation}
% By the definition of $\cT_i$, it is easy to see the equivalence of $\psi$ and $\hat{\psi}$.% $\{i\cup j\}\in\cK_\phi$, meaning $\bigwedge_{i\in\cV}(\bigwedge_{j\in\cT_i}\psi_{ij})$ and $\bigwedge_{\nu\in\cK_\phi}\psi_{\nu}$ are equivalent formulas, which leads to the equivalence of $\psi$ and $\hat{\psi}$. 

\subsubsection{Iterative algorithm}
For simplicity, we drop the time argument and introduce an iteration index as a superscript in the trajectory notation, e.g., $\bm{z}_i^k$ ($\bm{z}_{ij}^k$) indicates a trajectory $\bm{z}_i(0)$ ($\bm{z}_{ij}(0)$) that is retrieved at the $k$th iteration of the following procedure. To initialize the procedure, we generate initial guesses on the agents' trajectories by solving
\begin{subequations}\label{eq:initial_local_problem}
\begin{align}
    &\operatorname*{Minimize}_{\substack{\bm{v}_i^0, \bm{z}_i^0}}\sum_{t=0}^{N-1}(\ell_i(z_i^0(t), v_i^0(t))) + V_{f,i}(z_i^0(N)) \label{eq:cost_intro} \\
     &\mathrm{s.t.~} z_i^0(t+1) = A_iz_i^0(t)+B_iu_i^0(t), \; t\in\N_{[0,N)}, 
     \label{eq:local_STL_dynamics}  
     \\ 
     &\;\;\;\;\;\; \bm{z}_i^0 \models \psi_i, \; \mathrm{with}\; z_i^0(0)=x_{0,i}, \label{eq:local_STL_initial}  
     % &z_i(0)=x_{0,i}\label{eq:initial_agent_i_dynamics}
\end{align}    
\end{subequations}
at $k=0$ for $i\in\N_{[1,M]}$. Note that by Assum. \ref{ass:global_problem_tight}, Problem \eqref{eq:initial_local_problem} is guaranteed to be feasible. 
After performing the above initialization routine, at iteration $k\geq 1$, 
only a subset of agents, denoted by $\cO_k{\subset} \cV$, are allowed to update their 
input sequences by performing an optimization, while the remaining agents retrieve their 
input sequences from the previous iteration $k-1\geq 0$. 
Roughly, the set $\cO_k\subset \cV$ is constructed so that any combination of its elements 
does not belong to a clique $\nu\in\cK_\phi$. 
Due to space limitations, we simply construct $\cO_k$ as a singleton, 
which only affects the number of agents' trajectories that can be optimized in parallel 
per iteration. 
For the graph $\cG=(\cV,\cE)$, with $\cV=\N_{[1,M]}$, 
$\cO_k=\mathrm{mod}(\cO_{k-1},M)+1$, for $k>1$, with $\cO_1=1$. 
We refer readers to \cite[Sec. V.B]{vlahakisECC2023} for a more efficient construction of 
$\cO_k$ enabling parallel computations at each iteration (see Sec. \ref{sec:example} for a numerical example).

At the $k$th iteration, with $k\geq 1$, if $i\notin \cO_k$, then, $\bm{z}_i^k=\bm{z}_i^{k-1}$ and $\bm{v}_i^k=\bm{v}_i^{k-1}$. Otherwise, the input sequence of the $i$th agent is updated by solving
\begin{subequations}\label{eq:kth_local_problem}
\begin{align}
    &\operatorname*{Minim.}_{\substack{\bm{v}_i^k, \bm{z}_i^k}}\sum_{t=0}^{N-1}(\ell_i(z_i^k(t), v_i^k(t))) + V_{f,i}(z_i^k(N))-\mu_{ij_k}^k \label{eq:local_problem_kth_iter} \\
     &\mathrm{s.t.~} z_i^k(t+1) = A_iz_i^k(t)+B_iv_i^k(t), \; t\in\N_{[0,N)},\\ 
     & \;\;\;\;\;\; \bm{z}_i^k \models \psi_i, \; \mathrm{with}\; z_i(0)=x_{0,i}, \label{eq:kth_agent_i_local_task}\\  
     &\;\;\;\;\;\; \rho^{\psi_{ij_k}}(\bm{z}_{ij_k}^k) \geq \mu_{ij_k}^k, \; j_k = \argmin_{j\in \cT_i} \{\rho^{\psi_{ij}}(\bm{z}_{ij}^{k-1})\}, \label{eq:least_violating_joint_task_a}\\
     &\;\;\;\;\;\; \mu_{ij_k}^k\geq \min\left(0,\rho^{\psi_{ij_k}}(\bm{z}_{ij_k}^{k-1})\right), \label{eq:least_violating_joint_task_b}\\
     &\;\;\;\;\;\; \rho^{\psi_{ij}}(\bm{z}_{ij}^k)\geq \min\left(0,\rho^{\psi_{ij}}(\bm{z}_{ij}^{k-1})\right), \; \forall j \in \cT_i\setminus j_k, \label{eq:robust_joint_tasks}
\end{align}    
\end{subequations}
where $\rho^{\psi_{ij}}(\bm{z}_{ij}^k)$ is the robustness function of the formula $\psi_{ij}$ evaluated over the trajectory $\bm{z}_{ij}^k$. Agent-$i$, with $i\in \cO_k$, by solving \eqref{eq:kth_local_problem}, retrieves an input sequence that guarantees 1) the satisfaction of the individual task $\psi_i$ (see constraint \eqref{eq:kth_agent_i_local_task}), 2) the improvement of the most violating (or least robust) joint task $\psi_{ij_k}$ (see constraints \eqref{eq:least_violating_joint_task_a}-\eqref{eq:least_violating_joint_task_b}), and 3) either improvement on or non-violation of the remaining joint tasks (see constraint \eqref{eq:robust_joint_tasks}). The inclusion of the $\min$ operator in the constraints \eqref{eq:least_violating_joint_task_b}-\eqref{eq:robust_joint_tasks} relaxes the satisfaction of joint tasks that have already been found to be satisfiable in previous iterations. This allows the algorithm to emphasize the satisfaction of joint tasks with the smallest robustness function verified on trajectories attained in previous iterations. Based on the following result, these key features are ensured to persist at each iteration. The algorithm may terminate if it exceeds a maximum number of iterations, denoted as $k_\mathrm{max}$ and defined by the designer. Alternatively, termination occurs when verifying the satisfiability of all joint tasks, i.e., when $\mu_{ij}^k \geq 0$ for all $i \in \mathcal{V}$, $j \in \mathcal{T}_i$, and some $k\leq k_\mathrm{max}$. Given that the algorithm is not required to operate online, $k_\mathrm{max}$ may be determined ad-hoc. The overall iterative procedure is summarized in Algorithm \ref{alg:control_syn}.

\begin{theorem}
    At each iteration $k\geq 1$, the optimization problem \eqref{eq:kth_local_problem} is feasible for all $i\in \cO_k$.
\end{theorem}
\begin{proof}
    Let $k=1$. The lower bounds in \eqref{eq:least_violating_joint_task_a}-\eqref{eq:robust_joint_tasks} 
    are defined over trajectories, $\bm{z}_i^{0}$, $i{\in}\cO_1$, obtained by solving \eqref{eq:initial_local_problem} 
    at $k{=}0$. Thus $\bm{z}_i^{0}$ satisfies the constraints 
    \eqref{eq:least_violating_joint_task_a}-\eqref{eq:robust_joint_tasks} 
    for all $i{\in}\cO_1\subset \cV$. 
    Moreover, the constraint in \eqref{eq:kth_agent_i_local_task} is satisfied by $\bm{z}_i^{0}$, 
    since \eqref{eq:initial_local_problem} is feasible by Assum. \ref{ass:global_problem_tight}. 
    Hence, $\bm{u}_i^1 = \bm{u}_i^{0}$ is a feasible solution of \eqref{eq:kth_local_problem} 
    at iteration $k=1$. Now, let $k{>}1$. Similarly, the lower bounds in 
    \eqref{eq:least_violating_joint_task_a}-\eqref{eq:robust_joint_tasks} are defined over trajectories, 
    $\bm{z}_i^{k-1}$, obtained by the solutions $\bm{u}_i^{k-1}$, $i{\in}\cO_k$, at iteration $k-1$. 
    Thus, $\bm{z}_i^{k-1}$, $i{\in}\cO_k$, satisfy the constraints 
    \eqref{eq:least_violating_joint_task_a}-\eqref{eq:robust_joint_tasks}. 
    Additionally, the constraint in \eqref{eq:kth_agent_i_local_task} is satisfied by 
    $\bm{z}_i^{k-1}$, $i\in\cO_k$, since it is retrieved by $\bm{u}_i^{k-1}$, 
    which is obtained by solving \eqref{eq:kth_local_problem} or \eqref{eq:initial_local_problem} 
    at some iteration $\kappa\in \N_{[0, k-1]}$. 
    Thus, $\bm{u}_i^k {=} \bm{u}_i^{k-1}$ is a feasible solution of \eqref{eq:kth_local_problem}
    at iteration $k>1$. 
    We have shown that \eqref{eq:kth_local_problem} is feasible for all $i\in\cO_k$ at $k{=}1$ 
    and any $k{>}1$, so the result follows.
\end{proof}

\begin{algorithm}
\caption{Iterative procedure for solving \eqref{eq:multi_agent_problem_tight}}
\label{alg:control_syn}
\begin{algorithmic}[1]
% \State Given problem \eqref{eq:multi_agent_problem_tight}
\State Compute $\cT_i$ and construct $\hat{\psi}_i$ by \eqref{eq:psi_hat}, for $i\in\N_{[1,M]}$
\State Solve \eqref{eq:initial_local_problem} and store $\bm{v}_i^0$, $\bm{z}_i^0$, for $i\in\N_{[1,M]}$
\State Construct $\cO_k$, for $k\in\N_{[1,k_{\mathrm{max}}]}$ 
% For Loop
\For{$k$ in $1:k_\mathrm{max}$}
    \State \textbf{if} $i\in\cO_k$ solve \eqref{eq:kth_local_problem}, store $\bm{v}_i^k$ and compute $\bm{z}_i^k$ 
     \State \textbf{if} $i\notin\cO_k$ update $\bm{v}_i^k=\bm{v}_i^{k-1}$ and $\bm{z}_i^k=\bm{z}_i^{k-1}$
    \State Construct $\bm{v}(0)$ from $\bm{v}_i^k$, $\bm{z}(0)$ from $\bm{z}_i^k$, $i\in\N_{[1,M]}$
    \State \textbf{if} $\rho^\psi(\bm{z}(0))\geq 0$ \textbf{go to} \ref{algo:return}
    % \EndIf
\EndFor
\State \textbf{return} $\bm{v}(0)$, $\bm{z}(0)$ \label{algo:return}
\end{algorithmic}
\end{algorithm}

\begin{remark}
    Algorithm \ref{alg:control_syn} yields a minimally violating solution within a designer-defined iteration limit. If a solution is obtained before reaching the maximum iterations, it is guaranteed to be feasible for \eqref{eq:multi_agent_problem_tight}. By Theorem \ref{thm:deterministic_problem}, a feasible solution for \eqref{eq:multi_agent_problem} can then be efficiently constructed.
\end{remark}

\section{Example}\label{sec:example}

% $\begin{cases}
% 0 & \mathrm{if } x \leq a \\
% \frac{{(x - a)^2}}{{(b - a)(c - a)}} & \mathrm{if } a < x \leq c \\
% 1 - \frac{{(b - x)^2}}{{(b - a)(b - c)}} & \mathrm{if } c < x \leq b \\
% 1 & \mathrm{if } x > b
% \end{cases}$

We consider a MAS of ten agents with global dynamics given by $x(t+1)=x(t)+u(t)+w(t)$, where $x(t)=(x_1(t),..,x_{10}(t))\in\R^{20}$, $u(t)=(u_1(t),..,u_{10}(t))\in\R^{20}$, and $w(t)=(w_1(t),..,w_{10}(t))\in\R^{20}$. Individual states are constrained by $x_i(t)\in\cX$, where $\cX$ is the workspace of MAS, confined by the dashed border shown in Fig. \ref{fig:trajectories}, individual inputs are constrained by $\|u_i(t)\|_\infty\leq 0.8$, and individual disturbances, $w_i(t)$, are Gaussian random vectors, independent time- and agent-wise, with zero mean and covariance, $Q_i=0.05I_2$, for all $t\in\N$ and $i\in\N_{[1,10]}$. The MAS is assigned a global specification $\phi=\bigwedge_{i=1}^{10}\phi_i\wedge \bigwedge_{\nu\in\cK_\phi}\phi_\nu$, with horizon $N=100$, where $\cK_\phi$ is the set of cliques, each shown in Fig. \ref{fig:graph} in different colors, and $\phi_i$, $\phi_\nu$, are tasks assigned to agent $i$, and the agents of the clique $\nu\in\cK_\phi$, respectively. We solve \eqref{eq:multi_agent_problem}, with cost functions $\ell_i(x_i(t),u_i(t))=\|u_i(t)\|_1$, $V_{f,i}(x_i(100))=0$, and probability level $\theta=0.70$. 

Let $\phi_i= \left(\square_{[0,100]}(\phi_i^\cX\wedge \neg \phi_i^{O_1}\wedge\neg \phi_i^{O_2}\wedge\neg \phi_i^{O_3})\right) \wedge \left(\lozenge_{[10,50]} \phi_i^{T_i}\right) \wedge\left(\lozenge_{[70,100]}\phi_i^{G_i}\right)$ be an individual task, which requires agent-$i$, starting from an initial position, $x_i(0)$, to pass through a region $T_i$ within the interval $\N_{[10,50]}$, and reach a destination $G_i$ within the interval $\N_{[70,100]}$, while always staying within workspace $\cX$ and avoiding obstacles $O_1$, $O_2$ and $O_3$. Regions $T_i$, $G_i$, $i\in\N_{[1,10]}$, and obstacles $O_1$, $O_2$, $O_3$, are shown in Fig. \ref{fig:trajectories}. For each region, $\cY=\{\cX,O_1,O_2,O_3,T_1,\ldots,T_{10},G_1,\ldots,G_{10}\}$, the associated STL formula $\phi^\cY_i=\pi^\cY_{i,1} \wedge \pi^\cY_{i,2} \wedge \pi^\cY_{i,3} \wedge \pi^\cY_{i,4}$,  %whereas for the obstacles regions, $\cY=\{O_1,O_2,O_3\}$, $\phi^\cY_i=\neg\pi^\cY_{i,1} \lor \neg\pi^\cY_{i,2} \lor \neg\pi^\cY_{i,3} \lor \neg\pi^\cY_{i,4}$, 
where $\pi^\cY_{i,j}:=((a_{i,j}^{\cY})^\intercal  x_i(t)+b_{i,j}^{\cY}\geq0)$ is a predicate associated with a halfspace that intersects with the entire $j$th facet of the polyhedron $\cY$, with $j\in\N_{[1,4]}$.

Let $\phi_\nu=\lozenge_{[0,100]}\left(\|C_\nu x_\nu(t)\|_\infty\leq 1\right)$, where $C_\nu={\tiny\begin{bmatrix}
    I&-I
\end{bmatrix}}$ if $|\nu|=2$ or $C_\nu={\tiny\begin{bmatrix}
    I& -I& 0\\0 & I& -I\\I& 0&-I
\end{bmatrix}}$ if $|\nu|=3$, be a joint task requiring agents in $\nu\in\cK_\phi$ to approach one another at least once within the horizon. All joint tasks $\phi_\nu$, $\nu\in\cK_\phi$, are shown in Fig. \ref{fig:graph}. Note that $\phi_\nu$, with, e.g., $\nu=(3,4)$, is written as $\phi_{34}=\lozenge_{[0,100]} (\pi_{34}^1\wedge \pi_{34}^2 )\lor (\pi_{34}^3\wedge \pi_{34}^4)$, where $\pi_{34}^1:=([1\;0][-I\;I]x_{34}(t)+1\geq 0)$, $\pi_{34}^2:=([0\;1][-I\;I]x_{34}(t)+1\geq 0)$, $\pi_{34}^3:=([1\;0][I\;-I]x_{34}(t)+1\geq 0)$, and $\pi_{34}^4:=([0\;1][I\;-I]x_{34}(t)+1\geq 0)$.

\begin{figure}[t!]
\centering
\begin{tikzpicture}[auto, node distance=.75cm, every loop/.style={},thick,main node/.style={circle,draw}]
  % Nodes
  \node[main node] (2) at (1,1) {2};
  \node[main node,right=of 2] (3) {3};
  \node[main node,right=of 3] (4) {4};
  \node[main node,right=of 4] (7) {7};
  \node[main node,right=of 7] (8) {8};
  \node[main node,right=of 8] (10) {10};
  \node[main node,below=of 3] (1) {1};
  \node[main node,right=of 1] (5) {5};
  \node[main node,right=of 5] (6) {6};
  \node[main node,right=of 6] (9) {9};
  % \foreach \i in {1,...,10}
  %   \node[main node] (N\i) at (32*\i:2-1) {\i};
  % Edges
  % \foreach \i/\j in {1/2,1/3, 2/3, 3/4, 4/5, 4/7, 1/5, 5/6, 6/8, 6/9, 7/8, 8/10, 9/10}
  %   \draw[] (\i) -- (\j);
  %  % \foreach \name/\x in {1/2, 1/3, 2/3, 3/4, 4/5, 4/7, 1/5, 5/6, 6/8, 6/9, 7/8, 8/10, 9/10}
   %      \node[circle, draw, minimum size=0.5cm] (\name) at (\x,0) {\name};
    \draw[red,line width=2pt](1)--(2);
    \node[red,right=of 2, xshift=-.55cm,yshift=-.45cm] (123) {$\phi_{123}$};
    \draw[red,line width=2pt](2)--(3);
    \draw[red,line width=2pt](1)--(3);
    \draw[orange,line width=2pt](3)--(4);
    \node[orange,right=of 3, xshift=-.75cm,yshift=.25cm] (34) {$\phi_{34}$};
    \draw[cyan,line width=2pt](1)--(5);
    \node[cyan,below=of 34, xshift=-.cm,yshift=-.08cm] (15) {$\phi_{15}$};
    \draw[magenta,line width=2pt](5)--(6);
    \node[magenta,right=of 15, xshift=-.1cm,yshift=.cm] (56) {$\phi_{56}$};
    \draw[brown,line width=2pt](4)--(7);
    \node[brown,right=of 34, xshift=-.1cm,yshift=.cm] (47) {$\phi_{47}$};
    \draw[black,line width=2pt](7)--(8);
    \node[black,right=of 47, xshift=-.1cm,yshift=.cm] (78) {$\phi_{78}$};
    \draw[gray,line width=2pt](6)--(8);
    \node[gray,right=of 123, xshift=2.1cm,yshift=-.1cm] (68) {\rotatebox{45}{$\phi_{68}$}};
    \draw[blue,line width=2pt](6)--(9);
    \node[blue,right=of 56, xshift=-.1cm,yshift=.cm] (69) {$\phi_{69}$};
    \draw[green,line width=2pt](8)--(10);
    \node[green,right=of 78, xshift=-.27cm,yshift=.cm] (810) {$\phi_{8\;10}$};
    \draw[yellow,line width=2pt](9)--(10);
    \node[yellow,right=of 68, xshift=-.2cm,yshift=-.cm] (910) {\rotatebox{45}{$\phi_{9\;10}$}};
     \draw[purple,line width=2pt](4)--(5);
     \node[purple,left=of 68, xshift=-.2cm,yshift=-.1cm] (45) {\rotatebox{270}{$\phi_{45}$}};
    % Define colors
    % \def\linkcolors{{"red","red","red","orange","purple","cyan","magenta","brown","black","gray","blue","green"}} % Define array of colors
    
    % % Draw links with different colors
    % \foreach \i/\j [count=\k] in {1/2,1/3,2/3,3/4,4/5,4/7,1/5,5/6,6/8,6/9,7/8,8/10,9/10} {
    %     \draw[\linkcolors[\k]] (\i) -- (\j); % Link between nodes with color from array
    % }
\end{tikzpicture}
\caption{Graph representing the couplings in formula $\phi$.}
\label{fig:graph}
\end{figure}
First, to transform \eqref{eq:multi_agent_problem} into the deterministic problem \eqref{eq:multi_agent_problem_tight}, we define the \textit{tighter} formula $\psi$ by \eqref{eq:global_psi}-\eqref{eq:all_predicates}. To perform the optimizations in \eqref{eq:all_predicates}, we construct $t$-PRSs, $E_i(t)$, $t\in\N_{[0,100]}$, for \eqref{eq:ith_error_dyn} at probability levels $\theta_i$, $i\in\N_{[1,10]}$, such that $\Pi_{i=1}^{10}\left(1-100(1-\theta_i)\right)\geq 0.7$, and select $\theta_i=1-\frac{1-0.7^{\frac{1}{10}}}{100}=0.9996$ $\forall i\in\N_{[1,10]}$, by homogeneity. Since $w_i(t)$ is zero-mean Gaussian with covariance $Q_i=0.05I_2$, $\mathscr{E}_{\theta_i}(\cD(w_i(t))=\{w_i|w_i^\intercal  Q_i^{-1}w_i\leq \chi_2^2(\theta_i)\}$, where $\chi_2^2$ is the chi-squared distribution of degree $2$, is an ellipsoidal CR for $w_i(t)$ at probability level $\theta_i$. Then, we compute $E_i(t)$, $t\in\N_{[0,100]}$, recursively by $E_i(t+1)=\bar{A}_iE_i(t)\oplus \mathscr{E}_{\theta_i}(\cD_{w_i})$, with $E_i(0)=\{0\}$. The closed-loop matrix $\bar{A}_i=I_2+K_i$, where we select $K_i=-0.5I_2$, for $i\in\N_{[1,10]}$. The sets $E_i(t)$, $t\in\N_{[0,100]}$, are shown in Fig. \ref{fig:tubes}. Note that, by Theorem \ref{thm:tubes}, our selections of $\theta_i$, $i\in\N_{[1,10]}$, guarantee that $\bm{E}=E(0)\times \cdots \times E(100)$, with $E(t)=E_1(t)\times \cdots \times E_{10}(t)$, $t\in\N_{[0,100]}$, is a PRT for \eqref{eq:MAS_error_dyn} at probability level $\Theta=0.7$.

\begin{figure}[h]
    \centering
    \includegraphics[width=.7\linewidth]{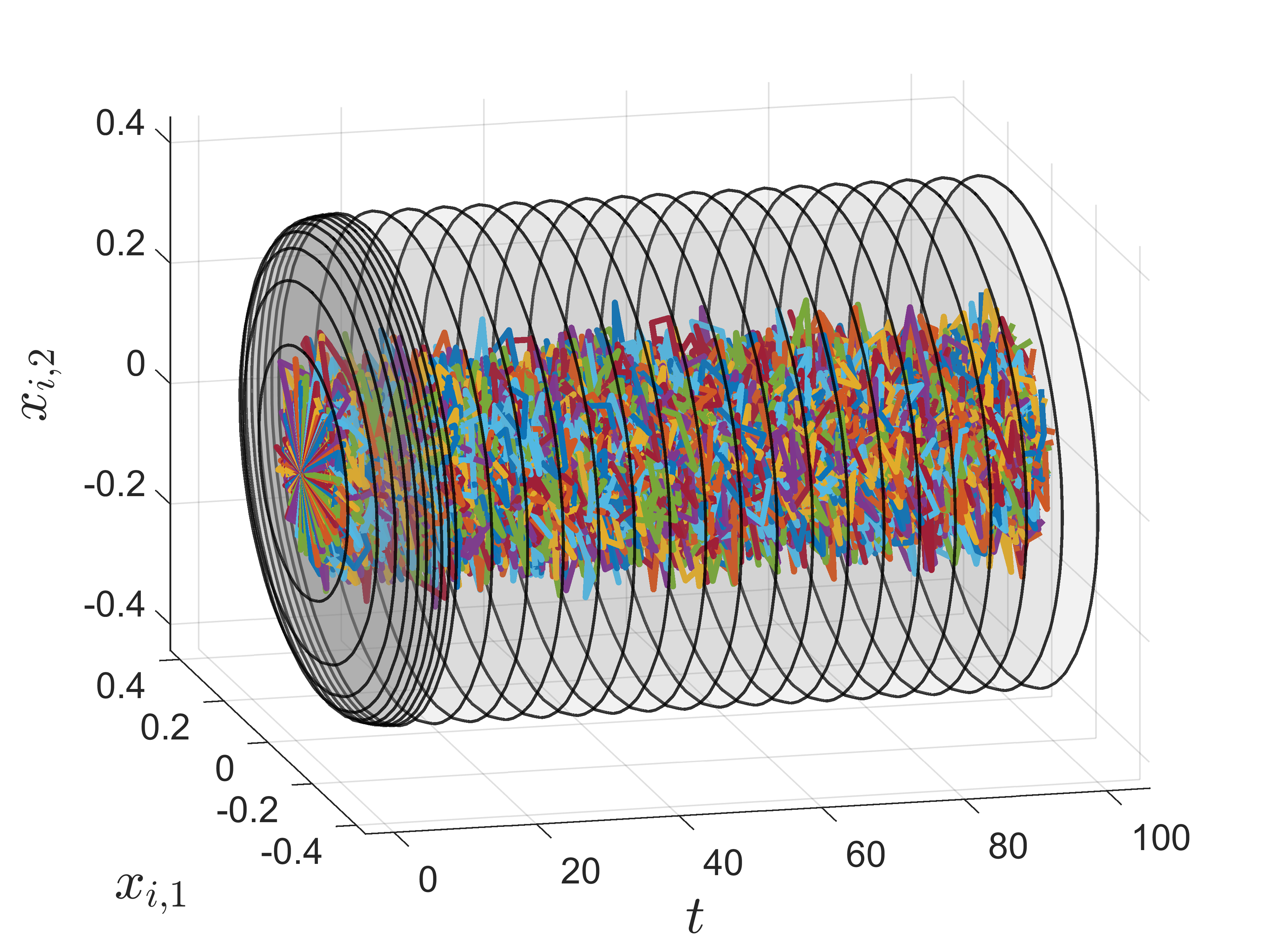}
    \caption{Trajectories for \eqref{eq:ith_error_dyn} lying in $\bm{E}_i=E_i(0)\times \cdots\times E_i(100)$, with probability $\Theta_i \geq  1-100(1-\theta_i)=0.965$.}
    \label{fig:tubes}
\end{figure}

% The sets $E_i(t)$, $t\in\N_{[0,100]}$, are the feasible domains of the optimization problems in \eqref{eq:tight1}-\eqref{eq:tight2}. The feasible domains in \eqref{eq:tight3}-\eqref{eq:tight4} are computed as follows. Let, e.g., $\nu=(1,2,3)\in\cK_\phi$. Then, $E_{123}(t)=E_1(t)\times E_2(t)\times E_3(t)$, $t\in\N_{[0,100]}$. After we have tightened all predicates in $\phi$ according to \eqref{eq:tighter_predicates}, 
After we have derived $\psi$, we formulate \eqref{eq:multi_agent_problem_tight} as an MILP using the YALMIP toolbox \cite{Yalmip}, in MATLAB. However, possibly due to the high dimensionality of the problem, the GUROBI solver \cite{gurobi} could not return a solution in a reasonable time. As a remedy, we implement the iterative optimization procedure presented in Section \ref{sec:distributed_control_synthesis}, after decomposing $\psi$ according to \eqref{eq:psi_hat}, based on the sets $\cT_1=\{(2,3),5\}$, $\cT_2=\{(1,3)\}$, $\cT_3=\{(1,2),4\}$, $\cT_4=\{3,5\}$, $\cT_5=\{1,4,6\}$, $\cT_6=\{5,8,9\}$, $\cT_7=\{4,8\}$, $\cT_8=\{6,7,10\}$, $\cT_9=\{6,10\}$, and $\cT_{10}=\{8,9\}$. By selecting sets $\cO_k$, $k>1$, as $\cO_1=\{1,4,6,10\}$, $\cO_2=\{3,5,7,9\}$, $\cO_3=\{8,9,2,5\}$, $\cO_4=\cO_1$, $\cO_5=\cO_2$, $\cO_6=\cO_3$, $\cO_7=\cO_1$, and so on, we run Algorithm \ref{alg:control_syn}, which terminates after seven iterations returning a multi-agent trajectory that satisfies the global STL task $\psi$. Computational aspects of Algorithm \ref{alg:control_syn} are depicted in Fig. \ref{fig:bothfigures} (left), while Fig. \ref{fig:bothfigures} (right) illustrates the evolution of the robustness function computed at the agent level during the execution of Algorithm \ref{alg:control_syn}. Fig. \ref{fig:trajectories} displays $10^3$ realizations of noisy multi-agent trajectories, illustrating also the impact of constraint tightening on the workspace $\cX$, obstacles $O_1$, $O_2$, $O_3$, and the regions $T_i$, $G_i$ ($i\in\N_{[1,10]}$). By evaluating the robustness function of the global STL specification $\phi$ for numerous noisy realizations, we see that $\phi$ is violated in less than $30\%$ of the total realizations, verifying Theorem \ref{thm:deterministic_problem}. %MATLAB scripts for this example can be accessed in \cite{githubcdc24}.

% \begin{figure}[htbp]
%     \centering
%     \includegraphics[width=.8\linewidth]{figures/7iterations_robustness29Feb24.eps}
%     \caption{}
%     \label{fig:robustness}
% \end{figure}

% \begin{figure}[htbp]
%     \centering
%     \includegraphics[width=.5\linewidth]{figures/hist_compute.eps}
%     \caption{Distribution of computing times per agent and iteration.}
%     \label{fig:hist}
% \end{figure}

\begin{figure}[ht]
  \centering
  \begin{subfigure}{0.26\columnwidth}
    \begin{tikzpicture}
	\begin{axis}[
        width=1in,
        height=1.8in,
		boxplot/draw direction = y,        
		x axis line style = {opacity=0},
            y axis line style = {opacity=0},
		axis x line* = bottom,
		axis y line = left,
		enlarge y limits,
		ymajorgrids,
		xticklabel style = {align=center, font=\small, rotate=60},
		xticklabels = {?},
		xtick style = {draw=none}, % Hide tick line
		ylabel = {Runtime (s)},
        ymode=log
	]
		\addplot [color=black, line width=1.1pt, dashed, forget plot]
  table[row sep=crcr]{%
1	23.731550925\\
1	36.8609245\\
};
\addplot [color=black, dashed, line width=1.1pt, forget plot]
  table[row sep=crcr]{%
1	4.3065434\\
1	9.638851025\\
};
\addplot [color=black, line width=1.1pt, forget plot]
  table[row sep=crcr]{%
0.9625	36.8609245\\
1.0375	36.8609245\\
};
\addplot [color=black, line width=1.1pt, forget plot]
  table[row sep=crcr]{%
0.9625	4.3065434\\
1.0375	4.3065434\\
};
\addplot [color=black, line width=1.1pt, forget plot]
  table[row sep=crcr]{%
0.925	9.638851025\\
0.925	23.731550925\\
1.075	23.731550925\\
1.075	9.638851025\\
0.925	9.638851025\\
};
\addplot [color=red, line width=1.3pt, forget plot]
  table[row sep=crcr]{%
0.925	17.1132848\\
1.075	17.1132848\\
};
\addplot [color=black, only marks, mark=+, mark options={solid, draw=red}, forget plot]
  table[row sep=crcr]{%
1	46.5631286\\
1	274.3\\
};
\end{axis}

% \begin{axis}[%
% width=5.833in,
% height=5.25in,
% at={(0in,0in)},
% scale only axis,
% xmin=0,
% xmax=1,
% ymin=0,
% ymax=1,
% axis line style={draw=none},
% ticks=none,
% axis x line*=bottom,
% axis y line*=left
% ] coordinates {};
% \end{axis}
\end{tikzpicture}
    %\caption{Computing times.}
    \label{fig:hist}
  \end{subfigure}
  \hfill 
  \begin{subfigure}{0.7\columnwidth}
    \definecolor{mycolor1}{rgb}{0.00000,0.90000,1.00000}%
    \definecolor{mycolor2}{rgb}{1.00000,0.00000,1.00000}%
    \definecolor{mycolor3}{rgb}{1.00000,1.00000,0.00000}%
    \definecolor{mycolor4}{rgb}{0.64000,0.16000,0.16000}%
\begin{tikzpicture}
\begin{axis}[%
    width=2.4in,
    height=1.8in,
    xmin=1, xmax=7,
    xlabel={Iterations},
    ymin=-16, ymax=2,
    ylabel={$\rho^{\hat{\psi}_i}$},
    axis background/.style={fill=white},
    axis x line*=bottom,
    xmajorgrids,
    ymajorgrids,
    legend columns=2,
    legend style={
        font=\tiny, 
        at={(1,0.02)},
        opacity=0.6,
        text opacity=1,
        anchor=south east
        },
    ylabel near ticks
]

\addplot [color=teal, dashed, line width=1.5pt, mark=x, mark options={solid, teal}]
  table[row sep=crcr]{%
1	-13\\
2	-8.99999999992162\\
3	-1.6029400029538e-12\\
4	-1.6029400029538e-12\\
5	-2.65742983174277e-12\\
6	2.57904808620424e-13\\
7	-9.80548975348938e-13\\
};
\addlegendentry{1};

\addplot [color=green, line width=1.5pt, mark=asterisk, mark options={solid, green!80!black}]
  table[row sep=crcr]{%
1	-2.29032503013078\\
2	-1.30163459012958\\
3	-1.30163459012958\\
4	-4.21884749357559e-15\\
5	-7.67719221528296e-12\\
6	-7.67719221528296e-12\\
7	-7.67719221528296e-12\\
};
\addlegendentry{2};

\addplot [color=blue, line width=1.5pt, mark=o, mark options={solid, blue}]
  table[row sep=crcr]{%
1	-8.99999999992162\\
2	-8.99999999992162\\
3	-1.30163459012958\\
4	-7.56150697611702e-12\\
5	-7.56150697611702e-12\\
6	-7.56150697611702e-12\\
7	-5.41700018175106e-12\\
};
\addlegendentry{3};

\addplot [color=mycolor1, line width=1.5pt, mark=square, mark options={solid, mycolor1}]
  table[row sep=crcr]{%
1	-14.0172484769578\\
2	-14.0172484769578\\
3	-14.0172484769564\\
4	-3.88922115296635\\
5	-3.88922115296635\\
6	-3.88922115296153\\
7	-1.48021594981174e-11\\
};
\addlegendentry{4};

\addplot [color=mycolor2, line width=1.5pt, mark=asterisk, mark options={solid, mycolor2}]
  table[row sep=crcr]{%
1	-13\\
2	-7.26166211363069\\
3	-0.862790110927389\\
4	-7.89113219212823e-12\\
5	-7.90301157849171e-12\\
6	-7.90301157849171e-12\\
7	-5.74318370638593e-12\\
};
\addlegendentry{5};

\addplot [color=mycolor3, line width=1.5pt, mark=o, mark options={solid, mycolor3}]
  table[row sep=crcr]{%
1	-9.40220864520256\\
2	-4.00819737696696\\
3	-0.862790110927389\\
4	2.77111666946439e-12\\
5	-2.70006239588838e-13\\
6	-3.66640051652212e-12\\
7	-1.5922818619174e-12\\
};
\addlegendentry{6};

\addplot [color=orange, dashed, line width=1.5pt, mark=x, mark options={solid, orange}]
  table[row sep=crcr]{%
1	-14.0172484769578\\
2	-14.0172484769578\\
3	-14.0172484769564\\
4	-3.88922115296635\\
5	-3.88922115296635\\
6	-3.88922115296153\\
7	-2.22044604925031e-16\\
};
\addlegendentry{7};

\addplot [color=violet, line width=1.5pt, mark=o, mark options={solid, violet}]
  table[row sep=crcr]{%
1	-9.00000000002568\\
2	-9.0000000000258\\
3	-3.6029245518477\\
4	-3.6029245518477\\
5	-3.49320572468059e-12\\
6	-3.49320572468059e-12\\
7	-3.49320572468059e-12\\
};
\addlegendentry{8};

\addplot [color=white!80!red, line width=1.5pt, mark=asterisk, mark options={solid, white!80!red}]
  table[row sep=crcr]{%
1	-9.40220864520256\\
2	-3.06767773216498\\
3	-0\\
4	1.60538249360798e-13\\
5	-1.26703092462321e-11\\
6	-1.26703092462321e-11\\
7	-1.26689769786026e-11\\
};
\addlegendentry{9};

\addplot [color=mycolor4, dashed, line width=1.5pt, mark=x, mark options={solid, mycolor4}]
  table[row sep=crcr]{%
1	-9.00000000002568\\
2	-9.0000000000258\\
3	-3.6029245518477\\
4	-3.6029245518477\\
5	-6.52322640348757e-12\\
6	3.56914497956495e-12\\
7	3.56914497956495e-12\\
};
\addlegendentry{10};

\end{axis}

\end{tikzpicture}%
    % \caption{Robustness function.}
    \label{fig:robustness}
  \end{subfigure}
  \caption{(Left) Compute time distribution for computations performed on an i7-1185G7 CPU at 3.00GHz. (Right) Robustness function of local formula $\hat{\psi}_i$ evaluated per iteration.}
  \label{fig:bothfigures}
\end{figure}

\vspace{-2em}
\begin{figure}[htbp]
\centering
\includegraphics[width=.45\textwidth]{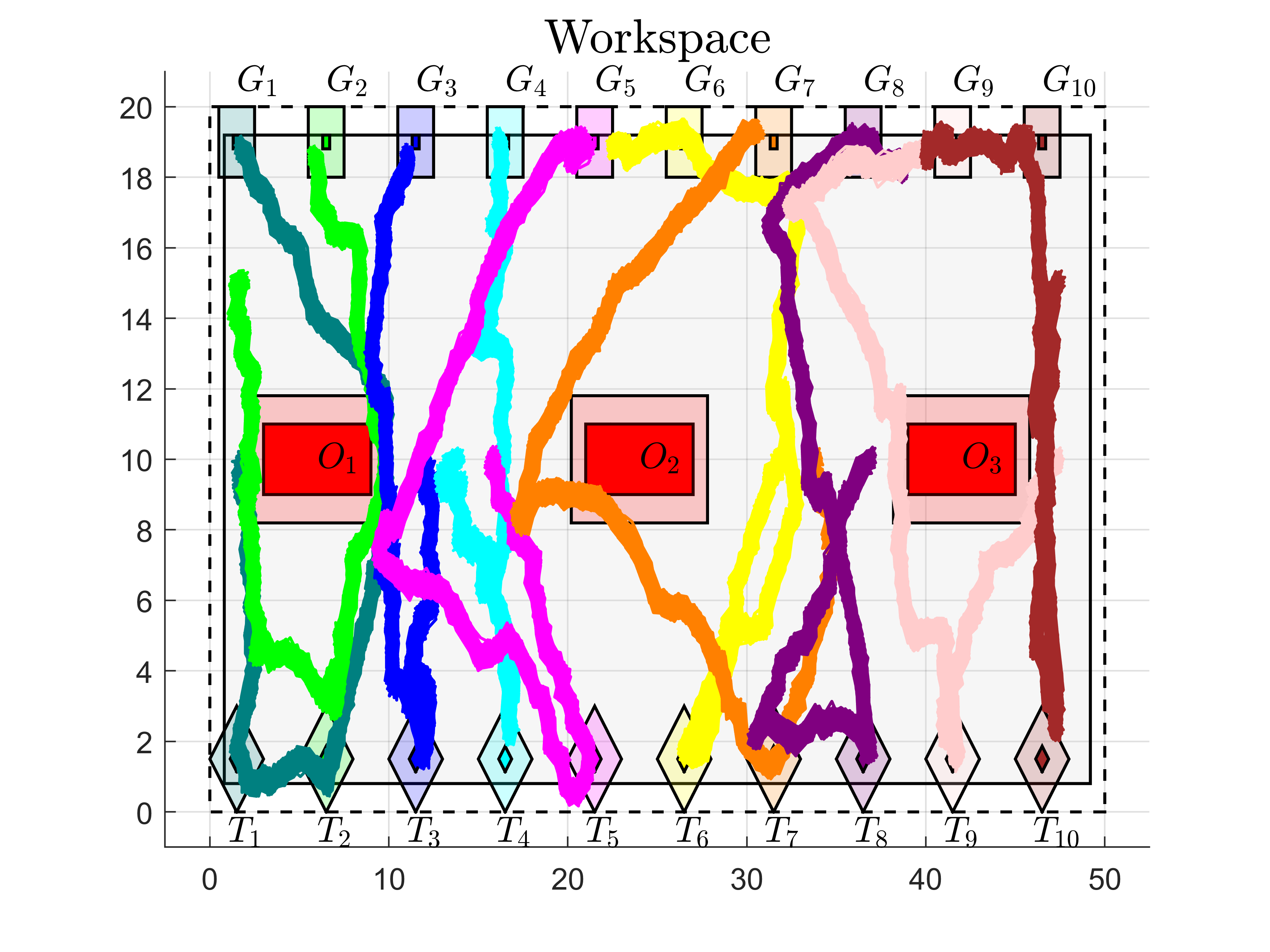}
\caption{$10^3$ realizations of the MAS with initial conditions $x_1(0)=(1.5,10)$, $x_2(0)=(1.5,15)$, $x_3(0)=(11.5, 10)$, $x_4(0)=(13.5, 10)$, $x_5(0)=(15.5, 10)$, $x_6(0)=(31.5, 10)$, $x_7(0)=(33.5, 10)$, $x_8(0)=(36.5, 10)$, $x_9(0)=(47,10)$, $x_{10}(0)=(47, 15)$, and effect of tightening in $\cX$, $O_1$, $O_2$, $O_3$, and $T_i$, $G_i$, $i\in\N_{[1,10]}$.% workspace $\cX$, obstacles $O_1$, $O_2$, $O_3$, and regions $T_i$, $G_i$, $i\in\N_{[1,10]}$. 
}
\label{fig:trajectories}
\end{figure}

\section{Conclusion}\label{sec:concl}

We have addressed the control synthesis of stochastic multi-agent systems under STL specifications formulated as probabilistic constraints. Leveraging linearity, we extract the stochastic component from the MAS dynamics for which we construct a PRT at probability level bounded by the specification probability. We use segments of this PRT to define a tighter STL formula and convert the underlying stochastic optimal control problem into a deterministic one with a tighter domain. Our PRT-based tightening reduces conservatism compared to approaches relying on the STL specification structure. To enhance scalability, we propose a recursively feasible iterative procedure where the resulting problem is decomposed into agent-level subproblems, reducing its complexity. Although our method fits a large-scale MAS setting, the conservatism introduced by the construction of PRTs increases with the specification horizon. This will be addressed in future work by utilizing more efficient, sampling-based techniques, avoiding union-bound arguments.

\bibliographystyle{IEEEtran}
\balance
\bibliography{biblio}

\end{document}